\documentclass{amsart}
\usepackage{amsmath}
\usepackage{amssymb,amsthm}
\usepackage{bm}
\usepackage[dvipdfmx]{graphicx}
\usepackage{mathrsfs}
\numberwithin{equation}{section}
\newtheorem{theorem}{Theorem}[section]
\newtheorem{remark}[theorem]{Remark}

\newtheorem{lem}[theorem]{Lemma}

\newtheorem{proposition}[theorem]{Proposition}

\title[Critical points of the Hartree-Fock functional]{Critical points of the discretized Hartree-Fock functional of connected molecules preserving structures of molecular fragments}
\author{Sohei Ashida}

\begin{document}
\maketitle

\begin{abstract}
In this paper a method to obtain a critical point of the discretized Hartree-Fock functional from an approximate critical point is given. The method is based on Newton's method on the Grassmann manifold. We apply Newton's method regarding the discretized Hartree-Fock functional as a function of a density matrix. The density matrix is an orthogonal projection in the linear space corresponding to the discretization onto a subspace whose dimension is equal to the number of electrons. The set of all such matrices are regarded as a Grassmann manifold. We develop a differential calculus on the Grassmann manifold introducing a new retraction (a mapping from the tangent bundle to the manifold itself) that enables us to calculate all derivatives. In order to obtain reasonable estimates, we assume that the basis functions of  the discretization are localized functions in a certain sense. As an application we construct a critical point of a molecule composed connecting several molecules using critical points of the Hartree-Fock functional corresponding to the molecules as the basis functions under several assumptions. By the error estimate of Newton's method we can see that the electronic structures of the molecular fragments are preserved.
\end{abstract}

\section{Introduction and statement of the result}\label{firstsec}

Construction of approximate eigenfunctions of electronic Hamiltonians is a fundamental problem in quantum chemistry. For $N,n\in\mathbb N$ the electronic Hamiltonian $H$ of $N$ electrons and $n$ nuclei is written as
$$H:=-\sum_{j=1}^N\Delta_{x_j}+\sum_{j=1}^NV(x_j)+\sum_{1\leq j<k\leq N}\frac{1}{|x_j-x_k|},$$
where
$$V(x):=-\sum_{l=1}^n\frac{Z_l}{|x-\bar x_l|}.$$
Here $x_j,\ 1\leq j\leq N$ (resp., $\bar x_l,\ 1\leq  l\leq n$) is the position of the $j$th electron (resp., $l$th nucleus), $\Delta_{x_j}$ is the Laplacian with respect to $x_j$, and $Z_l\in \mathbb N,\ 1\leq l\leq n$ is the atomic number of $l$th nucleus. The Hamiltonian $H$ acts on functions of the coordinates of electrons.

Actually, an electron has an internal state called spin. The spin does not affect the framework of the present result. However, when we consider connection of molecules, it seems quite unnatural to consider spin-independent functions, because it is known to be important for formation of chemical bonds that two electrons with different spins can occupy the same spatial function. Therefore, we introduce the spin of electrons here. The Hilbert space of spin internal state is $\mathbb C^2$. Thus the Hilbert space of a state of an electron is the tensor product $L^2(\mathbb R^3;\mathbb C)\otimes \mathbb C^2$. Let $\begin{pmatrix} 1\\ 0\end{pmatrix}$ and $\begin{pmatrix} 0\\ 1\end{pmatrix}$ be a basis of $\mathbb C^2$ and $L^2(\mathbb R^3;\mathbb C^2)$ be a Hilbert space equipped with the inner product $\int\varphi_1^*(x)\varphi_2(x)dx+\int\tilde\varphi_1^*(x)\tilde\varphi_2(x)dx$ for $\begin{pmatrix} \varphi_1\\ \tilde\varphi_1\end{pmatrix}, \begin{pmatrix} \varphi_2\\ \tilde\varphi_2\end{pmatrix}\in L^2(\mathbb R^3;\mathbb C^2)$. Then $L^2(\mathbb R^3;\mathbb C)\otimes \mathbb C^2$ is isomorphic to $L^2(\mathbb R^3;\mathbb C^2)$ by the mapping
$$L^2(\mathbb R^3;\mathbb C)\otimes \mathbb C^2\ni\varphi\otimes \begin{pmatrix} 1\\ 0\end{pmatrix}+\tilde\varphi\otimes\begin{pmatrix} 0\\ 1\end{pmatrix}\mapsto\begin{pmatrix} \varphi\\ \tilde\varphi\end{pmatrix}\in L^2(\mathbb R^3;\mathbb C^2).$$
In fact, we do not need tensor product and vector valued functions to introduce the spin. Let $\mathcal S:=\left\{\frac{1}{2},-\frac{1}{2}\right\}$ be a set composed of two elements. Let $L^2(\mathbb R^3\times\mathcal S;\mathbb C)$ be a Hilbert space equipped with the inner product
$$\langle\psi_1,\psi_2\rangle:=\sum_{\omega\in\mathcal S}\int_{\mathbb R^3}\psi_1^*(x,\omega)\psi_2(x,\omega)dx,$$
for $\psi_1,\psi_2\in L^2(\mathbb R^3\times\mathcal S;\mathbb C)$.
Then $L^2(\mathbb R^3;\mathbb C^2)$ is isomorphic to $L^2(\mathbb R^3\times\mathcal S;\mathbb C)$ by the mapping
$$L^2(\mathbb R^3;\mathbb C^2)\ni\begin{pmatrix} \varphi\\ \tilde\varphi\end{pmatrix}\mapsto \varphi(x)\alpha(\omega)+\tilde\varphi(x)\beta(\omega)\in L^2(\mathbb R^3\times\mathcal S;\mathbb C),$$
where $\alpha(\omega)$ and $\beta(\omega)$ are functions of $\omega\in\mathcal S$ such that $\alpha\left(\frac{1}{2}\right)=1$, $\alpha\left(-\frac{1}{2}\right)=0$, $\beta\left(\frac{1}{2}\right)=0$ and $\beta\left(-\frac{1}{2}\right)=1$. Thus in order to introduce the spin we have only to introduce the spin variable $\omega$ and replace $\int_{\mathbb R^3}$ by $\sum_{\omega\in\mathcal S}\int_{\mathbb R^3}$ when we integrate functions. This is a standard way to introduce the spin in quantum chemistry. Hereafter, we omit the notation and write as if there are not spin variables. We denote $L^2(\mathbb R^3\times\mathcal S;\mathbb C)$ and the Sobolev space $H^2(\mathbb R^3\times\mathcal S;\mathbb C)$ by $\mathcal H$ and $\tilde {\mathcal H}$ respectively.

The motivation of this paper is as follows. Electronic structures of different molecules correspond to eigenfunctions of different Hamiltonians $H$ in which the numbers of nuclei and electrons, nuclear positions $\bar x_l$ and atomic numbers $Z_l$ are different. However, from the observations of chemical experiments we know that the electronic structures of different molecules are not mutually irrelevant. When we synthesize a molecule from smaller molecules, the synthesized molecule is composed of parts corresponding to the original smaller molecules, and the electronic structures of these parts would be almost the same as those of the original molecules. Mathematically rigorous justification of this fact would be rather difficult, because in the eigenvalue problems of the Hamiltonians $H$ we consider eigenfunctions that are spreading around the whole molecule. Thus if the numbers of nuclei and electrons, nuclear positions and atomic numbers are different, the eigenvalue problems are treated as completely different problems. This approach is missing something about electronic structures of molecules. Intuitively, the eigenfunction of the synthesized molecule would be obtained by cutting and pasting local structures of molecular fragments. However, mathematically rigorous theoretical results in this direction can not be found in literature. We execute this idea under the framework of the Hartree-Fock approximation in this paper. More precisely, we prove that there exists a critical point of the Hartree-Fock functional of density matrices the corresponding wave function of which is close to a critical point of the Hartree-Fock functional corresponding to the molecular fragments under discretization by basis functions.

Let us introduce the Hartree-Fock functional. Let $\Phi=(\varphi_1,\dots,\varphi_N),\ \varphi_j\in \tilde{\mathcal H}$ be a tuple of functions. We define the Slater determinant $\Psi$ of $\Phi$ by
$$\Psi(x_1,\dots,x_N):=\sum_{\sigma\in S_N}(\mathrm{sgn}\, \sigma)\varphi_1(x_{\sigma(1)})\dotsm\varphi_N(x_{\sigma(N)}),$$
where $S_N$ is the symmetric group and $\mathrm{sgn}\, \sigma$ is the signature of $\sigma$. Then the Hartree-Fock functional is defined by
\begin{align*}
E(\Phi)&:=\langle \Psi,H\Psi\rangle\\
&=\sum_{j=1}^N\langle\varphi_j,h\varphi_j\rangle+\frac{1}{2}\sum_{j=1}^N\sum_{k=1}^N\int\int\varphi_j^*(x)\varphi_j(x)\frac{1}{|x-y|}\varphi_k^*(y)\varphi_k(y)dxdy\\
&\qquad-\frac{1}{2}\sum_{j=1}^N\sum_{k=1}^N\int\int\varphi_j^*(x)\varphi_k(x)\frac{1}{|x-y|}\varphi_k^*(y)\varphi_j(y)dxdy,
\end{align*}
where $h:=-\Delta+V(x)$ and we omitted the range $\mathbb R^3$ of the integral. Hereafter, when we omit the ranges of integrals, we assume that the ranges are $\mathbb R^3$.
Critical points $\Phi$ of $E(\Phi)$ under the constraints $\langle\varphi_j,\varphi_k\rangle=\delta_{jk}$ give approximations of eigenfunctions of $H$, and they are used also for further approximations.
Our purpose in this paper is to obtain a critical point of $E(\Phi)$ from an approximate critical point under discretization and apply the result to connected molecules to show existence of electronic structures of molecules preserving those of molecular fragments. A critical point of $E(\Phi)$ under the constraints $\langle\varphi_j,\varphi_k\rangle=\delta_{jk}$ satisfies the Euler-Lagrange equation
\begin{equation}\label{myeq1.0.0}
\mathcal F(\Phi)\varphi_j=\sum_{k=1}^N\epsilon_{jk}\varphi_k,\ j=1,\dots,N,
\end{equation}
where $\epsilon_{jk}\in\mathbb C,\ 1\leq j,k\leq N$ compose an Hermitian matrix and $\mathcal F(\Phi)$ is defined by
\begin{align*}
(\mathcal F(\Phi)\psi)(x):=&(h\psi)(x)+\sum_{j=1}^N\left(\int\frac{1}{|x-y|}\varphi_j^*(y)\varphi_j(y)dy\right)\psi(x)\\
&-\sum_{j=1}^N\left(\int\frac{1}{|x-y|}\varphi_j^*(y)\psi(y)dy\right)\varphi_j(x),
\end{align*}
and called Fock operator. Throughout this paper we denote by $(a_{jk})$ the matrix whose components are $a_{jk}$. After a unitary change
\begin{equation}\label{myeq1.0}
\varphi_j^{\mathrm{New}}=\sum a_{jk}\varphi_{k},
\end{equation}
by an appropriate $N\times N$ unitary matrix $(a_{jk})$, the new functions $(\varphi_1^{\mathrm{New}},\dots,\varphi_N^{\mathrm{New}})$ are orhonormal and satisfy the Hartree-Fock equation
$$\mathcal F(\Phi)\varphi_j=\epsilon_{j}\varphi_j,\ j=1,\dots,N,$$
with some real numbers $(\epsilon_1,\dots,\epsilon_N)$.

Actually, we consider a discretized functional. Let $\{\phi_1,\dots,\phi_{\nu}\},\ \phi_j\in\tilde{\mathcal H}, \nu\in\mathbb N, \nu>N$ be an orthonormal basis of some finite-dimensional subspace of $\mathcal H$ which includes a function very close to a critical point of $E(\Phi)$. Assume that $(\varphi_1,\dots,\varphi_N)\in\bigoplus_{j=1}^N\mathcal H$ is expressed as linear combinations of $\phi_k$:
\begin{equation}\label{myeq1.1}
\varphi_j=\sum_{k=1}^{\nu}c_{jk}\phi_k,\ j=1,\dots,N,
\end{equation}
where $c_{jk}\in\mathbb C$. Set $p_{jk}:=\sum_{l=1}^Nc_{lj}c^*_{lk}$. The matrix $P:=(p_{jk})$ is called a density matrix. Substituting \eqref{myeq1.1} into $E(\Phi)$ and using the density matrix $P$, the Hartree-Fock functional $E(\Phi)$ is rewritten as follows.
$$\mathcal E(P):=\sum_{j,k}h_{kj}p_{jk}+\frac{1}{2}\sum_{j,k,l,m}p_{jk}p_{lm}([kj|ml]-[kl|mj]),$$
where $h_{kj}:=\langle \phi_k,h\phi_j\rangle$ and $[kj|lm]:=\int\int\phi_k^*(x)\phi_j(x)\frac{1}{|x-y|}\phi_l^*(y)\phi_m(y)dxdy$. Here the ranges of summation with respect to $j,k,l$ and $m$ are $\{1,\dots,\nu\}$. Hereafter, when the ranges are not designated as in $\max_j$ or $\sum_j$, we assume that the ranges are $\{1,\dots,\nu\}$. If $(\varphi_1,\dots,\varphi_N)$ in \eqref{myeq1.1} is orthonormal, the corresponding density matrix $P$ is an orthogonal projection matrix of rank $N$. The set of all $\nu\times\nu$ orthogonal projection matrices of rank $N$ forms the Grassmann manifold $Gr(N,\nu)$ (cf. Section \ref{secondsec}). Thus we seek critical points of $\mathcal E(P)$ regarded as a function on $Gr(N,\nu)$.

For a critical point $P$ of $\mathcal E(P)$ there exist $N$ $\nu$-dimensional vectors $\mathbf c_j=(c_{j1},\dots,\newline c_{j\nu})^T,\ j=1,\dots,N$ such that $\{\bm c_1,\dots,\bm c_N\}$ is an orthonormal basis of $\mathrm{Ran}\, P$. The corresponding functions $\varphi_j$ constructed by \eqref{myeq1.1} from such $\bm c_j$ are called occupied orbitals. The tuple $\Phi:=(\varphi_1,\dots,\varphi_N)$ of occupied orbital is a solution to the Euler-Lagrange equation \eqref{myeq1.0.0} with some Hermitian matrix $(\epsilon_{jk})$ under the discretization, that is, the occupied orbitals span a direct sum of some eigenspaces of $\mathcal F(\Phi)$ regarded as an operator (matrix) on $\mathcal L(\phi_1,\dots,\phi_{\nu})$, where $\mathcal L(\phi_1,\dots,\phi_{\nu})$ is the linear subspace of $\mathcal H$ spanned by $\phi_1,\dots,\phi_{\nu}$. Since the dimension of $\mathcal L(\phi_1,\dots,\phi_{\nu})$ is $\nu$, there also exist other $\nu-N$ orthonormal eigenfunctions $\{\varphi_{N+1},\dots,\varphi_{\nu}\}$ of $\mathcal F(\Phi)$. These functions are called unoccupied orbitals. Usually the occupied orbitals are  assumed to be associated with $N$ lowest eigenvalues of $\mathcal F(\Phi)$. In this paper we shall also call the elements of the basis $\{\phi_1,\dots,\phi_{\nu}\}$ molecular orbitals, because they are regarded as occupied or unoccupied orbitals of molecular fragments in applications. Intuitively, in our result we assume that the bases $\{\phi_1,\dots,\phi_N\}$ and $\{\phi_{N+1},\dots,\phi_{\nu}\}$ are approximate occupied orbitals and approximate unoccupied orbitals respectively in a certain sense.

Our method is based on Newton's method. We construct a critical point of $\mathcal E(P)$ starting from the initial density matrix $P^0:= \mathrm{diag}\, (1,\dots,1,0,\dots,0)$, where $\mathrm{diag}\, (a_1,\dots,a_{\nu})$ denotes the diagonal matrix whose diagonal elements are $a_1,\dots,a_{\nu}$ and the first $N$ diagonal elements of $P^0$ are $1$. Actually, we apply the Kantorovi\v c semi-local convergence theorem. Since the Grassmann manifold is not a linear space, we use a local diffeomorphism called a retraction from a tangent space of the Grassmann manifold to the manifold itself (cf. Section \ref{secondsec}). For the Kantorovi\v c theorem we need estimates of derivatives of $\mathcal E(P)$. In order to obtain reasonable estimates we assume that $\{\phi_1,\dots,\phi_{\nu}\}$ are localized in the sense below. In quantum chemistry localized molecular orbitals are constructed practically from a critical point $\Phi=(\varphi_1,\dots,\varphi_N)$ of $E(\Phi)$ by a unitary matrix $(a_{jk})$ as in \eqref{myeq1.0} to obtain molecular orbitals displaying chemical bonds between atoms in the molecule (see e.g. \cite{Je}). The assumptions are needed so that the estimates do not depend linearly or worse on the sizes of molecules and the Hessian matrix of the functional becomes invertible in Newton's method.

We use weights $w_{jk},\ 1\leq j,k\leq \nu$ in order to express localization of molecular orbitals.  Intuitively, we assume $w_{jk}\sim |q_j-q_k|^s$ for some $s>1$ when $|q_j-q_k|$ is large enough, where $q_j\in\mathbb R^3$ is the reference point of the molecular orbital $\phi_j$, that is, $|\phi_j(x)|$ can have relatively large values only in a bounded region including $q_j$ and decays as $x$ leaves away from the region. 
We assume that $w_{jk}$ satisfy the following.
\begin{itemize}
\item [{\textbf (W)}] The matrix $(w_{jk})$ is symmetric, $w_{jk}\geq 1$, for any $j,k$, and the following conditions hold. 
\begin{itemize}
\item[(i)] $\max_j\sum_{k}w_{jk}^{-1}\leq 1$.
\item[(ii)]$w_{jk}^{-1}w_{kl}^{-1}\leq w_{jl}^{-1}$.
\end{itemize}
\end{itemize}

We set $[jk|lm]:=\int\int\phi_j^*(x)\phi_k(x)\frac{1}{|x-y|}\phi_l^*(y)\phi_m(y)dxdy$. We assume the following for integrals involving $\phi_j$.
\begin{itemize}
\item [{\textbf (LMO)}] There exist $\nu\times \nu$ symmetric matrices $(v_{jl}), (u_{jl})$ and constants $\tilde C, \hat C,\check C>0$ such that $ v_{jl} \geq 1$, $u_{jl}>0$ for any $j,l$ and the following conditions hold. 
\begin{itemize}
\item[(i)] There exists a constant $0<\tilde\epsilon<1$ such that $|[jk|lm]|\leq\tilde\epsilon  v_{jl}^{-1}w_{jk}^{-1}w_{lm}^{-1}$, when $\{j,k\}\neq \{l,m\}$ in addition to 
$j\neq k$ or $l\neq m$.
\item[(ii)] $\max_{l}\sum_j v_{jl}^{-1}\leq  1$.
\item[(iii)] $|[jk|lm]|\leq u_{jl}^{-1} w_{jk}^{-1} w_{lm}^{-1}$ for any $j,k,l,m$.
\item[(iv)] $\max_j\sum_{l}u^{-1}_{jl}\leq \tilde C$.
\item[(v)] $|[jk|lm]|\leq \hat Cw_{jk}^{-1} w_{lm}^{-1}$ for any $j,k,l,m$.
\item[(vi)] $\max_j\sum_k|\langle\nabla\phi_j,\nabla\phi_k\rangle|<\check C$.
\end{itemize}
\end{itemize}

We also need estimates of interactions between the orbitals. Set $J_o:=\{1,\dots,N\}$ and $J_u:=\{N+1,\dots,\nu\}$.

\begin{itemize}
\item[{\textbf (OI)}]
\begin{itemize}
\item[(i)] There exists a constant $0<\epsilon<1$ such that
\begin{align*}
\max_{k\in J_u}\sum_{j\in J_o}|\langle\phi_k,\mathcal F(\Phi^0)\phi_j\rangle|&\leq\epsilon,\\
\max_{j\in J_o}\sum_{k\in J_u}|\langle\phi_k,\mathcal F(\Phi^0)\phi_j\rangle|&\leq\epsilon,
\end{align*}
where $\Phi^0:=(\phi_1,\dots,\phi_N)$.
\item[(ii)] There exists a constant $\delta>0$ such that
\begin{align*}
\max_{j\in J_o}\sum_{\substack{k\in J_o \\ k\neq j}}|\langle\phi_j,\mathcal F(\Phi^0)\phi_k\rangle|&\leq\delta,\\
\max_{j\in J_u}\sum_{\substack{k\in J_u \\ k\neq j}}|\langle\phi_j,\mathcal F(\Phi^0)\phi_k\rangle|&\leq\delta.
\end{align*}
\item[(iii)] There exists a constant $\gamma>0$ such that for any pair of $j\in J_o$ and $k\in J_u$ we have
$$\langle\phi_k,\mathcal F(\Phi^0)\phi_k\rangle-\langle\phi_j,\mathcal F(\Phi^0)\phi_j\rangle-\langle kj||kj\rangle>\gamma,$$
where $\langle jl||km\rangle:=[jk|lm]-[jm|lk]$.
\end{itemize}
\end{itemize}

\begin{remark}\label{gaprem}
\begin{itemize}
\item[(i)] The condition (OI) (i) means that $\phi_j,\ j\in J_o$ are close to occupied orbitals in the sense that $\mathcal F(\Phi^0)\phi_j$ is close to a linear combination of $\{\phi_1,\dots,\phi_N\}$, and thus $\langle\phi_k,\mathcal F(\Phi^0)\phi_j\rangle$ is close to $0$ for $\phi_k,\ k\in J_u$.
\item[(ii)] The condition (OI) (ii) means that the interaction between $\phi_j$ and $\phi_k,\ k\neq j$ is somewhat small.
\item[(iii)] The condition (OI) (iii) means that the gap between the energies of $\phi_k$ and $\phi_j$ which are expected to be an approximate unoccupied and occupied orbital respectively is greater than the interaction by the Coulomb potential between the orbitals which is a driving force of the conversion from $\phi_j$ to $\phi_k$.
\end{itemize}
\end{remark}

Finally we assume a condition for the interactions between nuclei and the orbitals.
\begin{itemize}
\item[{\textbf (NI)}] There exist positive numbers $\breve u_{jkl}$ for $1\leq j,k\leq \nu,\ 1\leq l\leq n$ and a constant $\breve C>0$ such that
\begin{equation}\label{myeq1.3}
\left|\int\frac{Z_l}{|x-\bar x_l|}\phi_j^*(x)\phi_k(x)dx\right| \leq \breve u_{jkl}^{-1}w_{jk}^{-1},
\end{equation}
and $\max_{j,k}\sum_l\breve u_{jkl}^{-1}\leq \breve C$.
\end{itemize}

Next let us introduce norms of matrices. We define a norm $\lVert A\rVert_{1,\infty}$ for a matrix $A=(a_{jk})\in \mathbb C^{\nu\times\nu}$ by
$$\lVert A\rVert_{1,\infty}:=\max\{\lVert A\rVert_{1},\lVert A\rVert_{\infty}\},$$
where
$$\lVert A\rVert_1:=\sup_{\bm c\neq 0}\frac{\lVert A\bm c\rVert_1}{\lVert \bm c\rVert_1},\ \lVert A\rVert_{\infty}:=\sup_{\bm c\neq 0}\frac{\lVert A\bm c\rVert_{\infty}}{\lVert \bm c\rVert_{\infty}}.$$
Here $\bm c=(c_1,\dots,c_{\nu})^T\in\mathbb R^{\nu}$, $\lVert\bm c\rVert_1:=\sum_{j=1}^{\nu}|c_j|$ and $\lVert\bm c\rVert_{\infty}:=\max_{1\leq j\leq\nu}|c_j|$. As is well-known the norms are given explicitly by
$$\lVert A\rVert_1=\max_k\sum_{j=1}^{\nu}|a_{jk}|,\ \lVert A\rVert_{\infty}=\max_j\sum_{k=1}^{\nu}|a_{jk}|.$$
One reason why the norms $\lVert \cdot\rVert_1$ and $\lVert \cdot\rVert_{\infty}$ are preferable is that $\lVert I_{\nu}\rVert_1=\lVert I_{\nu}\rVert_{\infty}=1$ for the $\nu\times \nu$ identity matrix $I_\nu$ independently of $\nu$. Note also that $\lVert A\tilde A\rVert_1\leq\lVert A\rVert_1\lVert \tilde A\rVert_1$ and $\lVert A\tilde A\rVert_{\infty}\leq\lVert A\rVert_{\infty}\lVert \tilde A\rVert_{\infty}$, and thus $\lVert A\tilde A\rVert_{1,\infty}\leq\lVert A\rVert_{1,\infty}\lVert \tilde A\rVert_{1,\infty}$ for $A,\tilde A\in\mathbb C^{\nu\times\nu}$.
We use the norm $\lVert A\rVert_{1,\infty}$  defined in the same way also for $A\in\mathbb C^{N\times (\nu-N)}$.

We define a density matrix $P^0\in \mathbb C^{\nu\times\nu}$ by $P^0:=\mathrm{diag}\, (1,\dots,1,0,\dots,0)$, where $1$ appears $N$ times. The following main result means that there exists a critical point of $\mathcal E(P)$ close to $P^0$ under several assumptions.
\begin{theorem}\label{mainthm}
Assume (LMO), (OI), (NI) and $\frac{\gamma}{2}-\delta-2\tilde\epsilon>0$. Set $c_*:=\frac{1}{\frac{\gamma}{2}-\delta-2\tilde\epsilon}$. For a constant $0<\hat\epsilon<1$ we define $L_{\hat \epsilon}:=C_{\hat \epsilon}+3D_{\hat \epsilon}$, where 
\begin{align*}
C_{\hat\epsilon}:=&6(\tilde C+\hat C+\breve C+\check C)(1+\hat \epsilon)^2(1-\hat \epsilon^2)^{-3}\\
&\cdot\{1+2\hat \epsilon(1+(1-\hat \epsilon^2)^{-1}(1+\hat \epsilon)(1+3\hat \epsilon))\\
&+2(1-\hat \epsilon^2)^{-2}(1+\hat \epsilon)^2\hat \epsilon^2)\},
\end{align*}
\begin{align*}
D_{\hat\epsilon}:=&2(\tilde C+\hat C)(1+\hat \epsilon)(1-\hat \epsilon^2)^{-2}\{1+(1-\hat \epsilon^2)^{-1}(1+\hat \epsilon)\hat \epsilon\}\\
&\cdot\{ 1+(1-\hat \epsilon^2)^{-1}(1+\hat \epsilon)(1+5\hat \epsilon)+4(1-\hat \epsilon^2)^{-2}(1+\hat \epsilon)^2\hat \epsilon^2\}.
\end{align*}
Suppose that there exists $0<\hat\epsilon<1$ such that $c_*^2\epsilon L_{\hat\epsilon}<\frac{1}{2}$ and $\hat\epsilon>\tau_*$, where $\tau_*:=\frac{1-(1-2\theta)^{1/2}}{c_*L_{\hat\epsilon}}=\frac{2c_*\epsilon}{1+(1-2\theta)^{1/2}}<2c_*\epsilon$ with $\theta:=c_*^2\epsilon L_{\hat\epsilon}$. Then for such $\hat \epsilon$ there exists a critical point $P^{\infty}$ of $\mathcal E(P)$ such that $\lVert P^{\infty}-P^0\rVert_{1,\infty}\leq \tau_*\{1+(1-\tau_*^2)^{-1}(1+\tau_*)^2\}$ which is unique in a suitable neighborhood of $P^0$.
\end{theorem}
\begin{remark}
(i) As $\epsilon$ becomes small, $c^2\epsilon L_{\hat\epsilon}$ and $\tau_*$ become small. Hence for $\epsilon$ small enough there exists $\hat\epsilon$ such that the conditions hold.

(ii) We introduced the conditions (LMO), (NI) and (OI) so that the constant $L_{\hat\epsilon}$ does not become very large as the size of the molecule increases and the Hessian matrix of the functional becomes invertible in Newton's method. Note that our estimate is still not completely independent of the size of the molecule. The constants $\tilde C$ in (LMO) and $\breve C$ in (NI) would depend on the size. It seems that in order to obtain completely size-independent estimates we need to control the electronic density and show that the electrostatic potentials by electrons and nuclei cancel out during the process of Newton's method.
\end{remark}

Our method is based on a local diffeomorphism called a retraction from a tangent space of the Grassmann manifold to the manifold itself. We introduce a new retraction which enables us to calculate any derivative of a function on the Grassmann manifold. Using the differential calculus provided by the retraction we estimate the constants in the Kantorovi\v c semi-local convergence theorem.

\section{Retraction and differential calculus on the Grassmann manifold}\label{secondsec}
In this section we consider analysis on the complex Grassmann manifold. In particular, we introduce a method to treat the Grassmann manifold locally as a linear space. The results in this section are independent of specific applications in their nature. The complex Grassmann manifold $Gr(N,\nu)$ is the set of all $N$-dimensional subspaces of $\mathbb C^{\nu}$, where $N\leq \nu$. There exists a one-to-one correspondence between the $N$-dimensional subspaces and the orthogonal projection matrices of rank $N$:
$$Gr(N,\nu)\simeq\{P\in \mathbb C^{\nu\times \nu}:P^*=P,\ P^2=P,\ \mathrm{tr}(P)=N\}.$$
The Grassmann manifold $Gr(N,\nu)$ is an $N(\nu-N)$-dimensional complex manifold. However, since we are not concerned with the analyticity, we regard it as $2N(\nu-N)$-dimensional real manifold. Let $M$ be a real manifold. A smooth mapping $R$ from the tangent bundle $TM$ onto $M$ is called a retraction if the followings hold (cf. \cite{AMS}):
\begin{itemize}
\item[(i)] $R_x(0_x)=x$, where $0_x$ denotes the zero element of $T_xM$ and $R_x$ is the restriction of $R$ to $T_xM$.
\item[(ii)] With the canonical identification $T_{0_x}T_xM\simeq T_xM$, $R_x$ satisfies
$$dR_x(0_x)=id_{T_xM},$$
where $dR_x(0_x)$ denotes the differential of $R_x$ at $0_x$, and $id_{T_xM}$ denotes the identity mapping on $T_xM$.
\end{itemize}

Let us consider the tangent space of $Gr(N,\nu)$. Since $Gr(N,\nu)$ is a subset of $\mathbb C^{\nu\times \nu}$ ($\mathbb C^{\nu\times \nu}$ can be regarded as a $2(\nu\times\nu)$-dimensional real manifold), tangent vectors are also identified with elements in $\mathbb C^{\nu\times\nu}$. We have $T_PGr(N,\nu)=\{\xi\in \mathbb C^{\nu\times\nu}:\xi^*=\xi,\ \xi=\xi P+P\xi\}$ (cf. \cite[proof of Proposition 2.1]{SI}).
Let $P\in Gr(N,\nu)$ and $\bm y_1,\dots,\bm y_N\in \mathbb C^{\nu}$ be an orthonormal basis of $\mathrm{Ran}\, P$. Then if we set $Y:=(\bm y_1 \dotsm \bm y_N)$, we obviously have $P=YY^*$, where $(\bm y_1 \dotsm \bm y_N)$ is the matrix with columns $\bm y_1,\dots,\bm y_N$. Let $\bm y_{N+1},\dots,\bm y_{\nu}\in \mathbb C^{\nu}$ be an orthonormal basis of $\mathrm{Ker}\, P$. Then we have $I_{\nu}-P=Y_{\perp}Y_{\perp}^*$, where $Y_{\perp}:=(\bm y_{N+1}\dotsm \bm y_{\nu})$ and $I_{\nu}$ is the $\nu\times\nu$ identity matrix.
Using these notations the set of matrices
\begin{equation}\label{myeq2.1}
\begin{split}
\{\eta_{jk}:=&\mathrm{sym}(YE_{jk}Y^*_{\perp}): 1\leq j\leq N,\ N+1\leq k\leq \nu\}\\
&\cup\{\hat \eta_{jk}:=\mathrm{sym}(Y(iE_{jk})Y^*_{\perp}): 1\leq j\leq N,\ N+1\leq k\leq \nu\}
\end{split}
\end{equation}
forms a basis of the tangent space $T_PGr(N,\nu)$ of $Gr(N,\nu)$ at $P$ (cf. \cite[proof of Proposition 2.1]{SI}), where $E_{jk}$ is the $N\times (\nu-N)$ matrix whose $(j,k-N)$-component is $1$ and the others are $0$, and $\mathrm{sym} (A):=\frac{1}{2}(A+A^*)$. (Note that we consider $Gr(N,\nu)$ regarded as a $2N (\nu-N)$-dimensional real manifold, and thus the tangent space is regarded as a real linear space with respect to multiplication by real numbers.) Thus the general tangent vector $\xi\in T_PGr(N,\nu)$ is given by $\xi=\mathrm{sym}(YBY^*_{\perp})$, where $B\in\mathbb C^{N\times(\nu-N)}$. If we set a block matrix $B'\in \mathbb C^{\nu\times\nu}$ by
$$B':=\frac{1}{2}\begin{pmatrix}
0 &B\\
B^* &0
\end{pmatrix},$$
we can also write $\xi=QB'Q^*$, where $Q:=(Y Y_{\perp})$.

When we regard the Grassmann manifold as the set of orthonormal projections, a retraction based on the QR decomposition was introduced by \cite{SI}. They defined the retraction by $\tilde R_P(\xi):=\mathrm{qf}((I_{\nu}+\xi)Y)(\mathrm{qf}((I_{\nu}+\xi)Y)^T$, where $\mathrm{qf}(A)$ denote the $Q$ factor of the QR decomposition of $A=QR$, that is, $\mathrm{qf}(A)=Q$. However, this retraction is not suitable for the present purpose, because derivatives at $\xi\neq0$ and derivatives of order greater than $1$ of the retraction are difficult to calculate. In this paper we introduce a new retraction suitable for calculations of the derivatives. We define a mapping $R_P: T_PGr(N,\nu)\to Gr(N,\nu)$ by 
\begin{align*}
R_P(\xi)&=(I_{\nu}+\xi)Y[Y^*(I_{\nu}+\xi^*)(I_{\nu}+\xi)Y]^{-1}Y^*(I_{\nu}+\xi^*)\\
&=(I_{\nu}+\xi)Y[I_N+Y^*\xi^*\xi Y]^{-1}Y^*(I_{\nu}+\xi^*),
\end{align*}
for $\xi\in T_PGr(N,\nu)$, where we used $Y^*(I_{\nu}+\xi^*)(I_{\nu}+\xi)Y=Y^*Y+Y^*\xi^*\xi Y=I_N+Y^*\xi^*\xi Y$ (note that $Y^*\xi Y=Y^*\xi^*Y=0$ for $\xi\in T_PGr(N,\nu)$, which follows from that \eqref{myeq2.1} is a basis of $T_PGr(N,\nu)$). Here note that since $Y^*\xi^*\xi Y$ is positive, $I_N+Y^*\xi^*\xi Y$ is invertible. Since $K(K^*K)^{-1}K^*$ is the orthogonal projection onto $\mathrm{Ran}\, K$ for a non-singular matrix $K$, we can see that $R_P(\xi)\in Gr(N,\nu)$. Since $YY^*=P$, the condition $R_P(0_P)=P$ is satisfied. Moreover, we can see that $dR_P(0_P)[\zeta]=\zeta YY^*+YY^*\zeta^*=\zeta P+P\zeta=\zeta$ for $\zeta\in T_{0_P}T_PGr(N,\nu)= T_PGr(N,\nu)$, where we used $\zeta^*=\zeta$. Therefore, the mapping $(P,\xi)\mapsto R_P(\xi)$ is a retraction on $Gr(N,\nu)$. The Fr\' echet derivatives of $R_P(\xi)$ are summarized as follows. (Note that $R_P(\xi)$ is defined formally for any $\xi\in \mathbb C^{\nu\times\nu}$, and therefore, we can consider the Fr\'echet derivative of $R_P(\xi)$ in the usual way.) Hereafter, we omit the subscript $P$ of $0_P$.

\begin{proposition}\label{derivatives}
For $\xi,\zeta_1,\zeta_2,\zeta_3\in T_PGr(N,\nu)$ we have
\begin{itemize}
\item[(i)]
$$dR_P(\xi;\zeta_1)=\zeta_1 YZ^{-1}X^*-XZ^{-1}Y^*(\zeta_1^*\xi+\xi^*\zeta_1)YZ^{-1}X^*+XZ^{-1}Y^*\zeta_1^*,$$
where $X:=(I_{\nu}+\xi)Y$ and $Z:=I_N+Y^*\xi^*\xi Y$. In particular $dR_P(0;\zeta_1)=\zeta_1$.

\item[(ii)] 
$$d^2R_P(\xi;\zeta_1,\zeta_2)=\tilde R(\xi;\zeta_1,\zeta_2)+\tilde R(\xi;\zeta_2,\zeta_1),$$
where
\begin{align*}
\tilde R(\xi;\zeta_1,\zeta_2):=&-\zeta_1YZ^{-1}Y^*(\zeta_2^*\xi+\xi^*\zeta_2)YZ^{-1}X^*+\zeta_1YZ^{-1}Y^*\zeta_2^*\\
&-XZ^{-1}Y^*\zeta_1^*\zeta_2YZ^{-1}X^*\\
&+XZ^{-1}Y^*(\zeta_1^*\xi+\xi^*\zeta_1)YZ^{-1}Y^*(\zeta_2^*\xi+\xi^*\zeta_2)YZ^{-1}X^*\\
&-XZ^{-1}Y^*(\zeta_1^*\xi+\xi^*\zeta_1)YZ^{-1}Y^*\zeta_2^*.
\end{align*}
In particular
\begin{equation}\label{myeq2.1.1}
\begin{split}
d^2R_{P}(0;\zeta_1,\zeta_2)&=\zeta_1 YY^*\zeta_2^*+ \zeta_2YY^*\zeta_1^*-YY^*\zeta_1^*\zeta_2YY^*-YY^*\zeta_2^*\zeta_1 YY^*\\
&=\zeta_1 P\zeta_2^*+ \zeta_2P\zeta_1^*-P\zeta_1^*\zeta_2P-P\zeta_2^*\zeta_1 P,
\end{split}
\end{equation}
and 
\begin{equation}\label{myeq2.1.2}
d^2R_P(0;\eta_{jk},\hat\eta_{jk})=0.
\end{equation}

\item[(iii)]
\begin{align*}
d^3R_P(\xi;\zeta_1,\zeta_2,\zeta_3)=&\hat R(\xi;\zeta_1,\zeta_2,\zeta_3)+\hat R(\xi;\zeta_1,\zeta_3,\zeta_2)+\hat R(\xi;\zeta_2,\zeta_1,\zeta_3)\\
&+\hat R(\xi;\zeta_2,\zeta_3,\zeta_1)+\hat R(\xi;\zeta_3,\zeta_1,\zeta_2)+\hat R(\xi;\zeta_3,\zeta_2,\zeta_1),
\end{align*}
where
\begin{align*}
\hat R(\xi;\zeta_1,\zeta_2,\zeta_3):=&-\zeta_1YZ^{-1}Y^*\zeta_2^*\zeta_3YZ^{-1}X^*\\
&+\zeta_1YZ^{-1}Y^*(\zeta_2^*\xi+\xi^*\zeta_2)YZ^{-1}Y^*(\zeta^*_3\xi+\xi^*\zeta_3)YZ^{-1}X^*\\
&-\zeta_1YZ^{-1}Y^*(\zeta_2^*\xi+\xi^*\zeta_2)YZ^{-1}Y^*\zeta_3^*\\
&+XZ^{-1}Y^*\zeta_1^*\zeta_2YZ^{-1}Y^*(\zeta_3^*\xi+\xi^*\zeta_3)YZ^{-1}X^*\\
&-XZ^{-1}Y^*\zeta_1^*\zeta_2YZ^{-1}Y^*\zeta_3^*\\
&+XZ^{-1}Y^*(\zeta_1^*\xi+\xi^*\zeta_1)YZ^{-1}Y^*\zeta_2^*\zeta_3YZ^{-1}X^*\\
&-XZ^{-1}Y^*(\zeta_1^*\xi+\xi^*\zeta_1)YZ^{-1}Y^*(\zeta_2^*\xi+\xi^*\zeta_2)YZ^{-1}\\
&\qquad \cdot Y^*(\zeta_3^*\xi+\xi^*\zeta_3)YZ^{-1}X^*\\
&+XZ^{-1}Y^*(\zeta_1^*\xi+\xi^*\zeta_1)YZ^{-1}Y^*(\zeta_2^*\xi+\xi^*\zeta_2)YZ^{-1}Y^*\zeta_3^*.
\end{align*}
\end{itemize}
\end{proposition}

\begin{proof}
Except for \eqref{myeq2.1.2}, all formulas are obtained by direct calculations. For the proof of \eqref{myeq2.1.2} we define $\zeta^a:=YBY^*_{\perp}$ and $\zeta^b:=Y_{\perp}B^*Y^*$ for $\zeta=\mathrm{sym}(YBY^*_{\perp})\in T_PGr(N,\nu)$. Then we obviously have $\zeta=\frac{1}{2}(\zeta^a+\zeta^b)$. Noting that $Y^*Y_{\perp}=Y_{\perp}^*Y=0$ and $\zeta_j$ always appears in the form $\zeta Y$ or $Y^*\zeta^*$ in the Fr\'echet derivative $d^kR_P(\xi;\zeta_1,\dots,\zeta_k)$, we can see that if at least one of $\zeta_j$ in $d^kR_P(\xi;\zeta_1,\dots,\zeta_k)$ is replaced by $\zeta_j^a$ as $d^kR_P(\xi;\zeta_1,\dots,\zeta_j^a,\dots,\zeta_k)$, the Fr\'echet derivative vanishes. Thus we obtain
$$d^kR_P(\xi;\zeta_1,\dots,\zeta_k)=2^{-k}d^kR_P(\xi;\zeta_1^b,\dots,\zeta_k^b).$$
Hence noting also that $\hat\eta_{jk}^b=-i\eta_{jk}^b$, $(\zeta^b)^*=\zeta^a$ and $\zeta^bP=\zeta^b$ we have by \eqref{myeq2.1.1}
$$d^2R_P(0;\eta_{jk},\hat\eta_{jk})=\frac{1}{4}d^2R_P(0;\eta_{jk}^b,\hat\eta_{jk}^b)=\frac{1}{4}(i\eta_{jk}^b\eta_{jk}^a-i\eta_{jk}^b\eta_{jk}^a+i\eta_{jk}^a\eta_{jk}^b-i\eta_{jk}^a\eta_{jk}^b)=0.$$
\end{proof}

The following proposition is a consequence of the form of derivatives of $R_P(\xi)$.
\begin{proposition}
$dR_P(\xi):T_PGr(N,\nu)\to T_PGr(N,\nu)$ is an isomorphism for any $\xi\in T_PGr(N,\nu)$.
\end{proposition}
\begin{proof}
We have only to prove that $dR_P(\xi)$ is injective. Assume that
\begin{equation}\label{myeq2.1.2.1}
dR_P(\xi;\zeta)=\zeta YZ^{-1}X^*-XZ^{-1}Y^*(\zeta^*\xi+\xi^*\zeta)YZ^{-1}X^*+XZ^{-1}Y^*\zeta^*=0,
\end{equation}
for $\zeta\in T_PGr(N,\nu)$. Recall that $\zeta$ is written as $\zeta=\mathrm{sym}(YBY_{\perp}^*)$ for some $B \in \mathbb C^{N\times(\nu-N)}$. From the form \eqref{myeq2.1} of the basis of $T_PGr(N,\nu)$ it follows that $P\tilde\xi Y=0$ for any $\tilde\xi\in T_PGr(N,\nu)$. Since $X= Y+\xi Y$, multiplying \eqref{myeq2.1.2.1} by $P$ from the left and using also $PY=Y$ we obtain
$$-YZ^{-1}Y^*(\zeta^*\xi+\xi^*\zeta)YZ^{-1}X^*+YZ^{-1}Y^*\zeta^*=0.$$
This equation yields
$$-\xi YZ^{-1}Y^*(\zeta^*\xi+\xi^*\zeta)YZ^{-1}X^*+\xi YZ^{-1}Y^*\zeta^*=0.$$
It follows form these two equations that
$$-XZ^{-1}Y^*(\zeta^*\xi+\xi^*\zeta)YZ^{-1}X^*+XZ^{-1}Y^*\zeta^*=0.$$
Combining this equation and \eqref{myeq2.1.2.1} we obtain $\zeta YZ^{-1}X^*=0$. Since $X^*=Y^*+Y^*\xi^*$, by a similar argument as above it follows from $\zeta YZ^{-1}X^*=0$ that $\zeta YZ^{-1}Y^*=0$. Multiplying the equation by $YZ$ from the right we have $\zeta Y=0$, which means that $Y_{\perp} B^*=0$. This can hold only if $B=0$, that is $\zeta=0$, which implies that $dR_P(\xi)$ is injective and completes the proof.
\end{proof}

\begin{remark}
Since $Gr(N,\nu)$ is not homeomorphic to a linear space, $R_P(\xi):T_PGr(N,\nu)\to Gr(N,\nu)$ is not surjective. Actually, if we write $Y=(\bm y_1\dotsm \bm y_N)$ by orthonormal vectors $\bm y_1,\dots,\bm y_N$, $\mathrm{Ran}\, R_P(\xi)$ does not include any vector $\bm c\neq 0$ such that $\bm c\in \mathcal L(\bm y_1,\dots,\bm y_N)^{\perp}$.
\end{remark}

Next we consider derivatives of functions on $Gr(N,\nu)$. We denote the tangent space $T_PGr(N,\nu)$ by $\mathcal X=\mathcal X_P$. We can identify $\mathcal X$ with $\mathbb C^{N\times(\nu-N)}$ by $\mathcal X\ni\xi\mapsto B\in \mathbb C^{N\times(\nu-N)}$, where the correspondence is given by $\xi=\mathrm{sym}(YBY^*_{\perp})$. We introduce a norm in $\mathcal X$ by $\lVert \xi\rVert_{\mathcal X}:=\lVert B\rVert_{1,\infty}$ for $\xi=\mathrm{sym}(YBY^*_{\perp})$. From a function $\tilde f(P):Gr(N,\nu)\to\mathbb R$ and a point $P\in Gr(N,\nu)$ we obtain a function $f(\xi):=\tilde f(R_P(\xi)):\mathcal X\to\mathbb R$. By the definition of the Fr\'echet derivative for a function $f(\xi):\mathcal X\to\mathbb R$, we can see that $f'(\xi)$ is an element in the dual space $\mathcal X'$ of $\mathcal X$ whose components are given by $df(\xi;\eta_{jk})$ and $df(\xi;\hat\eta_{jk})$, $1\leq j\leq N,\ N+1\leq k\leq \nu$. We can identify $f'(\xi)\in\mathcal X'$ with $\tilde B=(\tilde b_{jk})\in\mathbb C^{N\times(\nu-N)}\simeq \mathbb R^{2(N\times(\nu-N))}\simeq\mathcal X$ by $\tilde b_{jk}:=df(\xi;\eta_{jk})+df(\xi;\hat\eta_{jk})i$. Thus the norm of $f'(\xi)$ is given by
\begin{equation}\label{myeq2.1.3}
\begin{split}
\lVert f'(\xi)\rVert_{\mathcal X'}=\lVert f'(\xi)\rVert_{\mathcal X}=\lVert \tilde B\rVert_{1,\infty}=\max\bigg\{&\max_j\sum_k\sqrt{(df(\xi;\eta_{jk}))^2+(df(\xi;\hat\eta_{jk}))^2},\\
&\max_k\sum_j\sqrt{(df(\xi;\eta_{jk}))^2+(df(\xi;\hat\eta_{jk}))^2}\bigg\}.
\end{split}
\end{equation}
Set $F(\xi):=f'(\xi)$. We can relabel the components $\mathrm{Re}\, F(\xi)_{jk},\ \mathrm{Im}\, F(\xi)_{jk}$ by $\mu=1,\dots, 2N(\nu-N)$ as $F(\xi)_{\mu}$. Moreover, if we suppose the labeling is executed so that $\mathrm{Re}\, F(\xi)_{jk}$ and $\mathrm{Im}\, F(\xi)_{jk},\ 1\leq j\leq N, N+1\leq k\leq \nu$ correspond to $F(\xi)_{2q-1}$ and $F(\xi)_{2q}$ respectively for some $1\leq q\leq N(\nu-N)$, we have 
\begin{equation*}
\begin{split}
F(\xi)_{2q-1}&=df(\xi;\eta_{jk}),\\
F(\xi)_{2q}&=df(\xi;\hat\eta_{jk}).
\end{split}
\end{equation*}
By definition $F'(\xi)$ is a real linear mapping from $\mathcal X\simeq\mathbb R^{2(N\times(\nu-N))}$ to itself. 

\begin{lem}\label{sumi}
For a function $f(P)$ of $P\in Gr(N,\nu)$ that has $m$th order derivative for $m\in\mathbb N$, a point $P'\in Gr(N,\nu)$, $\xi,\zeta_1,\dots,\zeta_m\in \mathcal X$ and $1\leq j\leq N,\ N+1\leq k\leq \nu$ we have
\begin{equation}\label{myeq2.4}
\begin{split}
&d^m(f(R_{P'}(\cdot)))(\xi;\zeta_1,\dots,\eta_{jk},\dots,\zeta_m)+d^m(f(R_{P'}(\cdot)))(\xi;\zeta_1,\dots,\hat\eta_{jk},\dots,\zeta_m)i\\
&\quad=NA_{\eta_{jk}^b}\{d^m(f(R_{P'}(\cdot)))(\xi;\zeta_1,\dots,\eta_{jk}^b,\dots,\zeta_m)\},
\end{split}
\end{equation}
where $d^m(f(R_{P'}(\cdot)))(\xi;\zeta_1,\dots,\zeta_m)$ is the $m$th Fr\'echet derivative of $f(R_{P'}(\tilde \xi))$ at $\xi$ in the direction $(\zeta_1,\dots,\zeta_m)$, $\eta_{jk}^b$ is the matrix defined in the proof of Proposition \ref{derivatives}, and $NA_{\tilde\zeta}\{d^m(f(R_{P'}(\cdot)))(\xi;\zeta_1,\dots,\tilde\zeta,\dots,\zeta_m)\}$ means that terms in $d^m(f(R_{P'}(\cdot)))(\xi;\newline \zeta_1,\dots,\tilde\zeta,\dots,\zeta_m)$ including $\tilde \zeta^*$ are replaced by $0$. In \eqref{myeq2.4} we assume that $\eta_{jk},\hat\eta_{jk}$ and $\eta_{jk}^b$ appear in the same place of $d^m(f(R_{P'}(\cdot)))(\xi;\zeta_1,\dots,\zeta_m)$.
\end{lem}
\begin{proof}
From the form of derivatives of $R_{P'}(\tilde\xi)$ we can see that the derivative
$$d^m(f(R_{\tilde P}(\cdot)))(\xi;\zeta_1,\dots,\tilde\zeta,\dots,\zeta_m)$$
contains terms linear with respect to $\tilde\zeta Y$ or $Y^*\tilde\zeta^*$. It follows from $\eta_{jk}^*=\eta_{jk}$, $\hat\eta_{jk}^*=\hat\eta_{jk}$, $\eta_{jk}+i\hat\eta_{jk}=\eta_{jk}^b$ and $Y^*\eta_{jk}^b=0$, that the sum of a term containing $Y^*\eta_{jk}^*$ and the corresponding term containing $iY^*\hat\eta_{jk}^*$ vanishes. On the other hand, the sum of a term containing $\eta_{jk}Y$ and the corresponding term containing $i\hat\eta_{jk}Y$ becomes a term containing $\eta_{jk}^bY$, which means the result.
\end{proof}

\section{proof of the main theorem}\label{thirdsec}
Our method is based on Newton's method. More precisely, we use the Kantorovi\v c semi-local convergence theorem which shows convergence of a sequence and existence of a solution without assuming the existence in advance. In Newton's method we consider an equation $F(\xi)=0$, where $F: \mathcal X \supset\mathcal D\to \mathcal X$ is a mapping from an open convex set $\mathcal D$ in a Banach space $\mathcal X$, and construct a sequence $\{\xi_m\}$ inductively by
$$\xi_{m+1}=\xi_m-F'(\xi_m)^{-1}F(\xi_m).$$
The following is the Kantorovi\v c semi-local convergence theorem.

\begin{proposition}[see e.g. \cite{Ze}]\label{Kant}
Suppose that:
\begin{itemize}
\item[(i)] The mapping $F$ is Fr\'echet differentiable on $\mathcal D$, and the derivative is Lipschitz continuous, i.e. there exists a constant $L>0$ such that
$$\lVert F'(\xi)-F'(\tilde\xi)\rVert\leq L\lVert \xi-\tilde\xi\rVert\quad \mathrm{for\ all}\ \xi,\tilde\xi\in \mathcal D.$$
\item[(ii)] For a fixed starting point $\xi_0\in\mathcal D$, the inverse $F'(\xi_0)^{-1}$ exists as a continuous linear operator on $\mathcal X$. The real numbers $c$ and $g$ are chosen so that
$$c\geq \lVert F'(\xi_0)^{-1}\rVert,\quad g\geq \lVert F'(\xi_0)^{-1}F(\xi_0)\rVert,$$
and $cgL<1/2$. Furthermore, we set $\theta:=cgL$, $\tau_*:=\frac{1-(1-2\theta)^{1/2}}{cL}$, $\tau_{**}:=\frac{1+(1-2\theta)^{1/2}}{cL}$, and $r:=\tau_*-g$.
\item[(iii)] The first approximation $\xi_1:=\xi_0-F'(\xi_0)^{-1}F(\xi_0)$ has the property that the closed ball $\bar U(\xi_1;r):=\{\zeta:\lVert\zeta-\xi_1\rVert\leq r\}$ lies within the domain of definition $\mathcal D$.
\end{itemize}
Then the equation $F(\xi)=0$ has a solution $\xi_*\in \bar U(\xi_1;r)$ and this solution is unique on $\bar U(\xi_0;\tau_{**})\cap\mathcal D$, i.e. on a suitable neighborhood of the initial point $\xi_0$.
\end{proposition}

Let $P^0$ be a diagonal matrix
$$P^0=\mathrm{diag}\, (1,\dots,1,0,\dots,0),$$
where the first $N$ diagonal elements are $1$. We apply Newton's method to $F(\xi):=(\mathcal E(R_{P^0}(\xi)))'$ with $\mathcal X=\mathcal X_{P^0}$, $\xi_0=0$ and $\mathcal D$ sufficiently close to $\xi_0$ with respect to the norm $\lVert\cdot\rVert_{\mathcal X}$. Under these settings we have the following lemma.

\begin{lem}\label{mainlem}
Under the assumptions of Theorem \ref{mainthm} the followings are true with the constants $c_*$ and $L_{\hat\epsilon}$ defined in Theorem \ref{mainthm}.
\begin{itemize}
\item[(1)] $$\lVert F(0)\rVert_{\mathcal X}\leq \epsilon,$$
\item[(2)] $$\lVert F'(0)^{-1}\rVert_{\mathcal L(\mathcal X)}\leq c_*,$$
\item[(3)] $$\lVert F'(\xi)-F'(\tilde\xi)\rVert_{\mathcal L(\mathcal X)}\leq L_{\hat\epsilon}\lVert \xi-\tilde\xi\rVert_{\mathcal X},$$
for $\xi,\tilde\xi\in \bar U(0;\hat\epsilon)$.
\end{itemize}
\end{lem}

\begin{proof}[proof of Theorem \ref{mainthm} under Lemma \ref{mainlem}]
By Lemma \ref{mainlem} and Proposition \ref{Kant} we can see that there exists a point $\xi^{\infty}$  such that $F(\xi^{\infty})=0$, i.e. a critical point of $\mathcal E(R_{P^0}(\xi))$ and obtain the estimate $\lVert\xi^{\infty}\rVert_{\mathcal X}\leq\tau_*$ with $\tau_*$ in Theorem \ref{mainthm}. Here note that we have $P^0=YY^*$ with $Y:=(\bm y_1 \dotsm \bm y_N)$, where $\bm y_j=(0,\dots,0,1,0,\dots,0)^T\in\mathbb C^{\nu}$ is the vector the only $j$th component of which is not zero.  Since we have
$$R_{P^0}(\xi^{\infty})-P^0=\xi^{\infty}YZ^{-1}X^*-YY^*(\xi^{\infty})^*\xi^{\infty}YZ^{-1}X^*+YY^*(\xi^{\infty})^*,$$
with the notations in Proposition \ref{derivatives} with $\xi$ replaced by $\xi^{\infty}$, the result follows easily if we set $P^{\infty}:=R_{P^0}(\xi^{\infty})$.
\end{proof}
Now it remains to prove Lemma \ref{mainlem}. For the proof we need the following lemma.

\begin{lem}\label{intnorm}
Assume (LMO). Let $A'=(\alpha_{jk})\in \mathbb C^{\nu\times\nu}$ be a matrix. Let us define $\Lambda_{jklm}:=[jk|lm]$,
Define also $T=(t_{jk}),\ \tilde T=(\tilde t_{jk}),\ \hat T=(\hat t_{jk})\in \mathbb C^{\nu\times\nu}$ by 
\begin{align*}
t_{jk}&:=(1-\delta_{jk})\sum_{l,m}(1-\delta_{\{j,k\}\{l,m\}})\Lambda_{jklm}\alpha_{lm},\\
\tilde t_{jk}&:=\sum_{l,m}\Lambda_{jklm}\alpha_{lm},\\
\hat t_{jk}&:=\sum_{l,m}\Lambda_{jmlk}\alpha_{lm},\\
\end{align*}
where
$$\delta_{\{j,k\}\{l,m\}}=\begin{cases}
1 &\quad\mathrm{if}\ \{j,k\}=\{l,m\}\\
0 &\quad\mathrm{otherwise}
\end{cases}.$$
Then we have
\begin{itemize}
\item[(a)] \begin{equation}\label{myeq3.1.1}
\lVert T\rVert_{1,\infty}\leq \tilde\epsilon\lVert A'\rVert_{1,\infty},
\end{equation}
\item[(b)] \begin{equation*}
\lVert \tilde T\rVert_{1,\infty}\leq \tilde C\lVert A'\rVert_{1,\infty},
\end{equation*}
\item[(c)] \begin{equation*}
\lVert \hat T\rVert_{1,\infty}\leq \hat C\lVert A'\rVert_{1,\infty}.
\end{equation*}
\end{itemize}
Moreover, the same estimate as \eqref{myeq3.1.1} holds if we replace $\Lambda_{jklm}$ in the definition of $t_{jk}$ by one of the values given by a permutation of $j,k,l,m$.
\end{lem}
\begin{proof}
(a) It follows from (LMO) (i), (ii) and (W) (i) that
\begin{align*}
\lVert T\rVert_{\infty}&\leq\max_j\sum_{k\neq j}\sum_{l,m}(1-\delta_{\{j,k\}\{l,m\}})\Lambda_{jklm}|\alpha_{l m}|\\
&\leq\tilde\epsilon\max_j\sum_{k\neq j}\sum_{l,m}(1-\delta_{\{j,k\}\{l,m\}})v_{jl}^{-1}w_{jk}^{-1}w_{lm}^{-1}|\alpha_{lm}|\\
&\leq\tilde\epsilon\lVert A'\rVert_{\infty},
\end{align*}
where the sums are accumulated in the order of $m,l,k$. The estimate for $\lVert T\rVert_{1}$ is similar.
Next let us consider the case in which $\Lambda_{jklm}$ is replaced by $\Lambda_{jlkm}$. Note first that it follows from $\{j,k\}\neq\{l,m\}$ and $j\neq k$ that $\{j,l\}\neq \{ k,m\}$ and either $j\neq l$ or $k\neq m$ as follows. If $\{j,l\}=\{k,m\}$, it follows from $j\neq k$ that $j=m$ and $l=k$, which contradicts $\{j,k\}\neq\{l,m\}$. Moreover, if $j=l$ and $k=m$, we again have $\{j,k\}=\{l,m\}$, which is a contradiction. Thus by (LMO) (i) and (W) (i) we have
\begin{align*}
\lVert T\rVert_{\infty}&\leq\max_j\sum_{k\neq j}\sum_{l,m}(1-\delta_{\{j,k\}\{l,m\}})\Lambda_{jlkm}|\alpha_{lm}|\\
&\leq\tilde\epsilon\max_j\sum_{k\neq j}\sum_{l,m}(1-\delta_{\{j,k\}\{l,m\}})v_{jk}^{-1}w_{jl}^{-1}w_{km}^{-1}|\alpha_{lm}|\\
&\leq\tilde\epsilon\lVert A'\rVert_{\infty},
\end{align*}
where the sums are accumulated in the order of $k,m,l$. The estimates for $\lVert T \rVert_{1}$ are similar to those for $\lVert T\rVert_{\infty}$. (The roles of $j$ and $k$ are interchanged.) The estimate for $\Lambda_{jmlk}$ is similar. Estimates for the other cases are essentially the same as those for the cases already mentioned.

(b) By (LMO) (iii), (iv) and (W) (i) we have
\begin{align*}
\lVert \tilde T\rVert_{\infty}&\leq\max_j\sum_{k}\sum_{l,m}\Lambda_{jklm}|\alpha_{lm}|\\
&\leq \max_j\sum_{k}\sum_{l,m}u_{jl}^{-1}w_{jk}^{-1}w_{lm}^{-1}|\alpha_{lm}|\\
&\leq \tilde C\lVert A'\rVert_{\infty},
\end{align*}
where the sums are accumulated in the order of $m,l,k$. The estimate for $\lVert \tilde T\rVert_{1}$ is similar.

(c) By (LMO)  (v) and (W) (i) we have
\begin{align*}
\lVert \hat T\rVert_{\infty}&\leq\max_j\sum_{k}\sum_{l,m}\Lambda_{jmlk}|\alpha_{lm}|\\
&\leq \hat C\max_j\sum_{k}\sum_{l,m}w_{jm}^{-1}w_{lk}^{-1}|\alpha_{lm}|\\
&\leq \hat C\lVert A'\rVert_{1},
\end{align*}
where the sums are accumulated in the order of $k,l,m$. The estimate for $\lVert \hat T\rVert_{1}$ is similar.
\end{proof}

We prove Lemma \ref{mainlem} calculating the derivatives of $\mathcal E(R_{P^0}(\xi))$ and estimating the derivatives. Here we mention that an expression of the Hessian of a function on the Grassmann manifold as a linear mapping was obtained by \cite{SI} and a somewhat abstract expression of the Hessian of $\mathcal E(P)$ at a local minimizer was obtained also by \cite{CKL} without using a retraction.
\begin{proof}[proof of Lemma \ref{mainlem}]
(1) Note first that if we denote by $\bm y_j:=(0,\dots,0,1,0,\dots,0)^T\in \mathbb C^{\nu},\ j=1,\dots,\nu$ the vector the only $j$th components of which is not zero, we can choose $Y:=(\bm y_1,\dots,\bm y_N)$ and $Y_{\perp}:=(\bm y_{N+1},\dots,\bm y_{\nu})$ as the matrices in Section \ref{secondsec}.  Thus the basis $\{\eta_{jk}\}\cup\{\hat \eta_{jk}\}$ of $T_{P^0}Gr(N,\nu)$ is expressed as $\eta_{jk}=\frac{1}{2}(\tilde E_{jk}+\tilde E_{kj})$ and $\hat\eta_{jk}=\frac{1}{2}(i\tilde E_{jk}-i\tilde E_{kj})$, where $\tilde E_{jk}$ is the $\nu\times\nu$ matrix whose $(j,k)$-component is $1$ and the others are $0$. Therefore, by Proposition \ref{derivatives} (i) and direct calculations we have
\begin{equation*}
\begin{split}
d(\mathcal E(R_{P^0}(\cdot)))(0;\eta_{jk})&=d\mathcal E(P^0;dR_{P^0}(0;\eta_{jk}))\\
&=\frac{1}{2}(\langle\phi_k,\mathcal F(\Phi^0)\phi_j\rangle+\langle\phi_j,\mathcal F(\Phi^0)\phi_k\rangle),
\end{split}
\end{equation*}
and
\begin{equation*}
d(\mathcal E(R_{P^0}(\cdot)))(0;\hat\eta_{jk})=\frac{1}{2}(i\langle\phi_k,\mathcal F(\Phi^0)\phi_j\rangle-i\langle\phi_j,\mathcal F(\Phi^0)\phi_k\rangle).
\end{equation*}
Thus \eqref{myeq2.1.3} and the condition (OI) (i) give $\lVert F(0)\rVert_{\mathcal X}=\lVert d(\mathcal E(R_{P^0}(\cdot))(0)\rVert_{\mathcal X}\leq\epsilon$.

(2) Let us denote by $K:T_{P^0}Gr(N,\nu)\to T_{P^0}Gr(N,\nu)$ the $2N(\nu-N)\times 2N(\nu-N)$ real block diagonal matrix whose diagonal blocks are 
\begin{align*}
&\begin{pmatrix}
d^2(\mathcal E(R_{P^0}(\xi)))(0;\eta_{jk},\eta_{jk}) &&d^2(\mathcal E(R_{P^0}(\xi)))(0;\eta_{jk},\hat\eta_{jk})\\
d^2(\mathcal E(R_{P^0}(\xi)))(0;\hat\eta_{jk},\eta_{jk}) && d^2(\mathcal E(R_{P^0}(\xi)))(0;\hat\eta_{jk},\hat\eta_{jk})
\end{pmatrix},\\
&\qquad1\leq j\leq N,\ N+1\leq k\leq \nu.
\end{align*}
We can immediately see that these diagonal blocks are those of $F'(0)$. Let us first estimate these diagonal blocks. We have
\begin{align*}
d^2\mathcal E(P^0;\eta_{jk},\eta_{jk})&=\frac{1}{4}([kj|kj]+[kj|jk]+[jk|kj]+[jk|jk])\\
&\quad-\frac{1}{4}([kj|kj]+[kk|jj]+[jj|kk]+[jk|jk])\\
&=\frac{1}{2}([kj|jk]-[kk|jj]).
\end{align*}
By \eqref{myeq2.1.1} we also have
\begin{align*}
d\mathcal E(P^0;d^2R_{P^0}(0;\eta_{jk},\eta_{jk}))=\frac{1}{2}\langle\phi_k,\mathcal F(\Phi^0)\phi_k\rangle-\frac{1}{2}\langle\phi_j,\mathcal F(\Phi^0)\phi_j\rangle.
\end{align*}
Thus (OI) (iii) gives
\begin{equation}\label{myeq3.1.2}
\begin{split}
&d^2(\mathcal E(R_{P^0}(\cdot)))(0;\eta_{jk},\eta_{jk})\\
&\quad=d^2\mathcal E(P^0;\eta_{jk},\eta_{jk})+d\mathcal E(P^0;d^2R_{P^0}(0;\eta_{jk},\eta_{jk}))\\
&\quad=\frac{1}{2}\langle\phi_k,\mathcal F(\Phi^0)\phi_k\rangle-\frac{1}{2}\langle\phi_j,\mathcal F(\Phi^0)\phi_j\rangle+\frac{1}{2}([kj|jk]-[kk|jj])\\
&\quad\geq \frac{\gamma}{2}.
\end{split}
\end{equation}
In the same way we have
\begin{equation}\label{myeq3.2}
\begin{split}
&d^2(\mathcal E(R_{P^0}(\cdot)))(0;\hat\eta_{jk},\hat\eta_{jk})\\
&\quad=\frac{1}{2}\langle\phi_k,\mathcal F(\Phi^0)\phi_k\rangle-\frac{1}{2}\langle\phi_j,\mathcal F(\Phi^0)\phi_j\rangle+\frac{1}{2}([kj|jk]-[kk|jj])\geq \frac{\gamma}{2}.
\end{split}
\end{equation}
Moreover, a direct calculation yields
\begin{align*}
d^2\mathcal E(P^0;\eta_{jk},\hat\eta_{jk})=&\frac{1}{4}(i[kj|kj]-i[kj|jk]+i[jk|kj]-i[jk|jk])\\
&-\frac{1}{4}(i[kj|kj]-i[kk|jj]+i[jj|kk]-i[jk|jk])=0.
\end{align*}
Since \eqref{myeq2.1.2} gives
$$d\mathcal E(P^0;d^2R_{P^0}(0;\eta_{jk},\hat\eta_{jk}))=0,$$
we obtain 
\begin{equation}\label{myeq3.3}
d^2(\mathcal E(R_{P^0}(\cdot)))(0;\eta_{jk},\hat\eta_{jk})=0.
\end{equation}
It follows from \eqref{myeq3.1.2}--\eqref{myeq3.3} that $\lVert K^{-1}\rVert_{\mathcal L(\mathcal X)}\leq\frac{2}{\gamma}$. 

As for the off-diagonal elements of $F'(0)$ we have
$$d^2\mathcal E(P^0;\eta_{jk},\zeta)+d^2\mathcal E(P^0;\hat\eta_{jk},\zeta)i=\frac{1}{2}\sum_{\substack{l\in J_o\\ m\in J_u}}(\langle jm||kl\rangle b_{lm}+\langle jl||km\rangle b_{lm}^*),$$
and
\begin{align*}
&d\mathcal E(P^0;d^2R_{P^0}(0;\eta_{jk},\zeta))+d\mathcal E(P^0;d^2R_{P^0}(0;\hat\eta_{jk},\zeta))i\\
&\quad=\frac{1}{2}\sum_{\substack{l\in J_o\\ m\in J_u}}b_{lm}(\delta_{jl}\langle \phi_m,\mathcal F(\Phi^0)\phi_k\rangle-\delta_{km}\langle\phi_j,\mathcal F(\Phi^0)\phi_l\rangle),
\end{align*}
where $\zeta=\mathrm{sym}(YBY_{\perp}^*),\ B=(b_{lm})$. Let us estimate the off-diagonal elements.
It follows from (OI) (ii) that
$$\max_{j\in J_o}\sum_{k\in J_u}\sum_{\substack{l\in J_o\\ m\in J_u}}|(1-\delta_{jl}\delta_{km})\delta_{jl}\langle\phi_m,\mathcal F(\Phi^0)\phi_k\rangle b_{lm}|<\delta\lVert B\rVert_{\infty},$$
and
$$\max_{j\in J_o}\sum_{k\in J_u}\sum_{\substack{l\in J_o\\ m\in J_u}}|(1-\delta_{jl}\delta_{km})\delta_{km}\langle\phi_j,\mathcal F(\Phi^0)\phi_l\rangle b_{lm}|<\delta\lVert B\rVert_{\infty},$$
where the sums are accumulated in the order of $l,k,m$ and $m,k,l$ respectively.
By \eqref{myeq3.1.1} we obtain
$$\max_{j\in J_o}\sum_{k\in J_u}\sum_{\substack{l\in J_o\\ m\in J_u}}|(1-\delta_{jl}\delta_{km})\langle jm||kl\rangle b_{lm}|\leq 2\tilde\epsilon\lVert B\rVert_{1,\infty}.$$
The same estimate holds for the term $\langle jl||km\rangle b_{lm}^*$. The estimates for the cases in which the maximum and the summation with respect to $j$ and $k$ are interchanged are similar. Thus if we denote by $\tilde K:T_{P^0}Gr(N,\nu)\to T_{P^0}Gr(N,\nu)$ the $2N(\nu-N)\times 2N(\nu-N)$ real matrix defined by $F'(0)-K$, we have $\lVert\tilde K\rVert_{\mathcal L(\mathcal X)}\leq \delta+2\tilde\epsilon$. Therefore, by the Neumann series $(K+\tilde K)^{-1}=K^{-1}\sum_{j=0}^{\infty}(-1)^j(K^{-1}\tilde K)^j$ and the assumption $\frac{\gamma}{2}-\delta-2\tilde\epsilon>0$ we can see that $K+\tilde K$ is invertible and $\lVert  (K+\tilde K)^{-1}\rVert_{\mathcal L(\mathcal X)}\leq\frac{2}{\gamma}\frac{1}{1-\frac{2}{\gamma}(\delta+2\tilde\epsilon)}=\frac{1}{\frac{\gamma}{2}-\delta-2\tilde\epsilon}=c_*$. This means $\lVert F'(0)^{-1}\rVert_{\mathcal L(\mathcal X)}\leq c_*$.

(3) As for the Lipschitz continuity of $F'(\xi)$ we have
\begin{align*}
\lVert F'(\tilde\xi)-F'(\xi)\rVert_{\mathcal L(\mathcal X)}&=\sup_{\lVert\zeta\rVert_{\mathcal X}=1}\lVert dF(\tilde\xi;\zeta)-dF(\xi;\zeta)\rVert_{\mathcal X}\\
&=\sup_{\lVert\zeta\rVert_{\mathcal X}=1}\left\lVert\int_0^1d^2F(\xi+(\tilde \xi-\xi)t;\zeta,\tilde\xi-\xi)dt\right\rVert_{\mathcal X}.
\end{align*}
Thus if we show $\lVert d^2F(\xi;\zeta,\tilde\zeta)\rVert_{\mathcal X}\leq L_{\hat\epsilon}\lVert \zeta\rVert_{\mathcal X}\lVert\tilde\zeta\rVert_{\mathcal X}$ for $\xi\in\bar U(0;\hat\epsilon)$, we obtain the Lipschitz continiuity. By Lemma \ref{sumi} this means that the inequality
$$\max_{j\in J_o}\sum_{k\in J_u}|NA_{\eta_{jk}^b}\{d^3(\mathcal E(R_{P^0}(\cdot)))(\xi;\eta_{jk}^b,\zeta,\tilde\zeta)\}|\leq L_{\hat\epsilon}\lVert \zeta\rVert_{\mathcal X}\lVert\tilde\zeta\rVert_{\mathcal X},$$
and the inequality in which the indices of summation $k$ and maximum $j$ are interchanged hold.

We define $\mathcal T(P):=\sum_{j,k}p_{jk}\langle \nabla\phi_k,\nabla\phi_j\rangle$, $\mathcal V(P):=\sum_{j,k}p_{jk}\langle \phi_k,V\phi_j\rangle$ and
\begin{align*}
\tilde {\mathcal G}(P)&:=\frac{1}{2}\sum_{j,k,l,m}p_{jk}p_{lm}[kj|ml],\\
\hat {\mathcal G}(P)&:=-\frac{1}{2}\sum_{j,k,l,m}p_{jk}p_{lm}[kl|mj].
\end{align*}
Then $\mathcal E(P)$ is rewritten as $\mathcal E(P)=\mathcal T(P)+\mathcal V(P)+\tilde {\mathcal G}(P)+\hat {\mathcal G}(P)$.
The third derivative of $\tilde {\mathcal G}(R_{P^0}(\xi))$ is given by
\begin{align*}
d^3(\tilde {\mathcal G}(R_{P^0}(\cdot)))(\xi;\eta_{jk}^b,\tilde\zeta,\hat \zeta)=&d\tilde {\mathcal G}(R_{P^0}(\xi);d^3R_{P^0}(\xi;\eta_{jk}^b,\tilde\zeta,\hat\zeta))\\
&+d^2\tilde {\mathcal G}(R_{P^0}(\xi);d^2R_{P^0}(\xi;\eta_{jk}^b,\tilde\zeta),dR_{P^0}(\xi;\hat\zeta))\\
&+d^2\tilde {\mathcal G}(R_{P^0}(\xi);d^2R_{P^0}(\xi;\eta_{jk}^b,\hat\zeta),dR_{P^0}(\xi;\tilde\zeta))\\
&+d^2\tilde {\mathcal G}(R_{P^0}(\xi);dR_{P^0}(\xi;\eta_{jk}^b),d^2R_{P^0}(\xi;\tilde\zeta,\hat\zeta)).
\end{align*}

Here we note that if a function $f$ is given by $f(B')=\sum_{j,k}d_{kj}\beta_{jk}$ for $B'=(\beta_{jk})\in\mathbb C^{\nu\times\nu}$, we have
$$f(A\eta_{jk}^b\tilde A)=\sum_{l,m}d_{ml}a_{lk}\tilde a_{jm},$$
for $A=(a_{jk})$ and $\tilde A=(\tilde a_{jk})$.
Hence we can estimate as
\begin{align*}
\max_{j\in J_o}\sum_{k\in J_u}|f(A\eta_{jk}^b\tilde A)|\leq \lVert D\rVert_{1,\infty}\lVert A\rVert_{1,\infty}\lVert\tilde A\rVert_{1,\infty},\\
\max_{k\in J_u}\sum_{j\in J_o}|f(A\eta_{jk}^b\tilde A)|\leq \lVert D\rVert_{1,\infty}\lVert A\rVert_{1,\infty}\lVert \tilde A\rVert_{1,\infty},
\end{align*}
where $D:=(d_{jk})$.
Therefore, if we write $\xi=Y^0BY_{\perp}^0$, $\tilde \zeta=Y^0\tilde BY_{\perp}^0$ and $\hat\zeta=Y^0\hat BY_{\perp}^0$, using Lemma \ref{intnorm} (b) and Proposition \ref{derivatives} we obtain
\begin{align*}
&\max_{j\in J_o}\sum_{k\in J_u}\left|NA_{\eta_{jk}^b}\left\{d\tilde {\mathcal G}(R_{P^0}(\xi);d^3R_{P^0}(\xi;\eta_{jk}^b,\tilde\zeta,\hat\zeta))\right\}\right|\leq \tilde C_{B}\lVert  \tilde B\rVert_{1,\infty}\lVert \hat B\rVert_{1,\infty},\\
&\max_{j\in J_o}\sum_{k\in J_u}\left|NA_{\eta_{jk}^b}\left\{d^2\tilde {\mathcal G}(R_{P^0}(\xi);d^2R_{P^0}(\xi;\eta_{jk}^b,\tilde\zeta),dR_{P^0}(\xi;\hat\zeta))\right\}\right|\leq \tilde D_{B}\lVert  \tilde B\rVert_{1,\infty}\lVert \hat B\rVert_{1,\infty},\\
&\max_{j\in J_o}\sum_{k\in J_u}\left|NA_{\eta_{jk}^b}\left\{d^2\tilde {\mathcal G}(R_{P^0}(\xi);d^2R_{P^0}(\xi;\eta_{jk}^b,\hat\zeta),dR_{P^0}(\xi;\tilde\zeta))\right\}\right|\leq \tilde D_{B}\lVert  \tilde B\rVert_{1,\infty}\lVert \hat B\rVert_{1,\infty},\\
&\max_{j\in J_o}\sum_{k\in J_u}\left|NA_{\eta_{jk}^b}\left\{d^2\tilde {\mathcal G}(R_{P^0}(\xi);dR_{P^0}(\xi;\eta_{jk}^b),d^2R_{P^0}(\xi;\tilde\zeta,\hat\zeta))\right\}\right|\leq \tilde D_{B}\lVert  \tilde B\rVert_{1,\infty}\lVert \hat B\rVert_{1,\infty},
\end{align*}
where $\tilde C_{B}>0$ and $\tilde D_B$ are given by
\begin{align*}
\tilde C_B:=&6\tilde C(1+\lVert B\rVert)^2(1-\lVert B\rVert^2)^{-3}\\
&\cdot\{1+2\lVert B\rVert(1+(1-\lVert B\rVert^2)^{-1}(1+\lVert B\rVert)(1+3\lVert B\rVert))\\
&+2(1-\lVert B\rVert^2)^{-2}(1+\lVert B\rVert)^2\lVert B\rVert^2)\},
\end{align*}
and
\begin{align*}
\tilde D_B&:=2\tilde C(1+\lVert B\rVert)(1-\lVert B\rVert^2)^{-2}\{1+(1-\lVert B\rVert^2)^{-1}(1+\lVert B\rVert)\lVert B\rVert\}\\
&\cdot\{ 1+(1-\lVert B\rVert^2)^{-1}(1+\lVert B\rVert)(1+5\lVert B\rVert)+4(1-\lVert B\rVert^2)^{-2}(1+\lVert B\rVert)^2\lVert B\rVert^2\}.
\end{align*}
Here we omitted the subscript $1,\infty$ of the norm.
As a result we obtain
\begin{equation}\label{myeq3.4}
\max_{j\in J_o}\sum_{k\in J_u}\left|NA_{\eta_{jk}^b}\left\{d^3(\tilde {\mathcal G}(R_{P^0}(\cdot)))(\xi;\eta_{jk}^b,\tilde\zeta,\hat \zeta)\right\}\right|\leq (\tilde C_B+3\tilde D_{B})\lVert  \tilde B\rVert_{1,\infty}\lVert \hat B\rVert_{1,\infty}.
\end{equation}

As for $\hat {\mathcal G}(P)$ using Lemma \ref{intnorm} (c) an argument similar to that of $\tilde {\mathcal G}(P)$ yields
\begin{equation}\label{myeq3.5}
\max_{j\in J_o}\sum_{k\in J_u}\left|NA_{\eta_{jk}^b}\left\{d^3(\hat {\mathcal G}(R_{P^0}(\cdot)))(\xi;\eta_{jk}^b,\tilde\zeta,\hat \zeta)\right\}\right|\leq (\hat C_B+3\hat D_{B})\lVert  \tilde B\rVert_{1,\infty}\lVert \hat B\rVert_{1,\infty},
\end{equation}
where $\hat C_B$ and $\hat D_B$ are defined by replacing $\tilde C$ in $\tilde C_B$ and $\tilde D_B$ by $\hat C$ respectively.
As for $\mathcal V(P)$ and $\mathcal T(P)$ using (NI) and (LMO) (vi) respectively we have
\begin{equation}\label{myeq3.6}
\max_{j\in J_o}\sum_{k\in J_u}\left|NA_{\eta_{jk}^b}\left\{d^3\mathcal V(\xi;\eta_{jk}^b,\zeta,\tilde\zeta)\right\}\right|\leq \breve C_B\lVert  \tilde B\rVert_{1,\infty}\lVert \hat B\rVert_{1,\infty},
\end{equation}
and
\begin{equation}\label{myeq3.7}
\max_{j\in J_o}\sum_{k\in J_u}\left|NA_{\eta_{jk}^b}\left\{d^3\mathcal T(\xi;\eta_{jk}^b,\zeta,\tilde\zeta)\right\}\right|\leq \check C_B\lVert  \tilde B\rVert_{1,\infty}\lVert \hat B\rVert_{1,\infty},
\end{equation}
where $\breve C_B$ and $\check C_B$ are defined by replacing $\tilde C$ in $\tilde C_B$ by $\breve C$ and $\check C$ respectively.
The estimates corresponding to \eqref{myeq3.4}--\eqref{myeq3.7} for the summation with respect to $j$ are obtained in the same way, which means the result.
\end{proof}

\section{Idea of the conditions of localized molecular orbitals}\label{fourthsec}
 The meaning of (OI) is explained in Remark \ref{gaprem}. In order to argue the validity of (LMO) we first assume that each $\phi_j(x)$ is localized near a reference point $q_j\in\mathbb R^3$, that is, $|\phi_j(x)|$ can have relatively large values only in a relatively small bounded region including $q_j$ and decays as $x$ leaves away from the region. The factor  $w_{jk}^{-1}$ (resp., $w_{lm}^{-1}$) in (LMO) (i), (iii) and (v) would come from the decay of $\phi_j^*(x)\phi_k(x)$ (resp., $\phi_l^*(y)\phi_m(y)$) as $|q_j-q_k|$ (resp., $|q_l-q_m|$) increases. Intuitively, we assume $w_{jk}^{-1}=|q_j-q_k|^{-s},\ s>1$.
 
 The factor $v_{jl}^{-1}$ in (LMO) (i) is explained as follows. Assume $j\neq k$. By Taylor's theorem $\frac{1}{|x-y|}$ is rewritten as
$$\frac{1}{|x-y|}=\frac{1}{|q_j-y|}-\int_0^1(x-q_j)\cdot\frac{tx+(1-t)q_j-y}{|tx+(1-t)q_j-y|^3}dt.$$
Since $j\neq k$, we have $\int\phi_j^*(x)\phi_k(x)dx=0$. Thus we can see that
\begin{equation}\label{myeq4.2}
\begin{split}
&\int\phi_j^*(x)\phi_k(x)\frac{1}{|x-y|}\phi_l^*(y)\phi_m(y)dxdy\\
&\quad=-\int_0^1\int\phi_j^*(x)\phi_k(x)(x-q_j)\cdot\frac{tx+(1-t)q_j-y}{|tx+(1-t)q_j-y|^3}\phi_l^*(y)\phi_m(y)dxdydt.
\end{split}
\end{equation}
Since $\phi_j^*(x)$ and $\phi_l^*(y)$ are localized near $q_j$ and $q_l$ respectively, the factor $\frac{tx+(1-t)q_j-y}{|tx+(1-t)q_j-y|^3}$ in \eqref{myeq4.2} would yield the decay factor $v_{jl}^{-1}\sim\frac{1}{|q_j-q_l|^s},\ s>1$. When $l\neq m$ similar estimates would hold replacing $q_j$ and $x$ by $q_l$ and $y$ respectively in \eqref{myeq4.2}.

Assume again $j\neq k$. Then the factor $\tilde \epsilon$ in (LMO) (i) would come from the smallness of $\phi_j^*(x)\phi_k(x)$, $\phi_l^*(y)\phi_m(y)$ or $\frac{tx+(1-t)q_j-y}{|tx+(1-t)q_j-y|^3}$, when $|q_j-q_k|$, $|q_l-q_k|$ or $|q_j-q_l|$ is large. Otherwise, it would come from the small overlap of components  of $\phi_j\phi^*_k$ and $\phi_l^*\phi_m$ with respect to the spherical harmonics centered at some point $\tilde q$ close to $q_j, q_k, q_l$. The Coulomb potential $\frac{1}{|x-y|}$ is expanded by the Laplace expansion (see e.g. \cite{Ja}[Section 3.6]) as
 \begin{equation*}
 \frac{1}{|x-y|}=4\pi\sum_{\lambda=0}^{\infty}\sum_{\mu=-\lambda}^{\lambda}\frac{1}{2\lambda+1}\frac{r_<^{\lambda}}{r_>^{\lambda+1}}Y^*_{\lambda\mu}(\theta_y,\varphi_y)Y_{\lambda\mu}(\theta_x,\varphi_x),
 \end{equation*}
where $(\theta_x,\varphi_x, r_x)$ and $(\theta_y,\varphi_y, r_y)$ are spherical coordinates of $x-\tilde q$ and $y-\tilde q$ respectively, and $r_>:=\max\{r_x,r_y\}, r_<:=\min\{r_x,r_y\}$. Thus we have
\begin{equation}\label{myeq4.1}
\begin{split}
&\int\phi_j^*(x)\phi_k(x)\frac{1}{|x-y|}\phi_l^*(y)\phi_m(y)dxdy\\
&\quad=4\pi\sum_{\lambda=0}^{\infty}\sum_{\mu=-\lambda}^{\lambda}\frac{1}{2\lambda+1}\int\int\bigg\{\frac{r_<^{\lambda}}{r_>^{\lambda+1}}\left(\int\int\phi_j^*(x)\phi_k(x)Y_{\lambda\mu}(\theta_x,\varphi_x)d\theta_xd\varphi_x\right)\\
&\hspace{140pt}\cdot\left(\int\int Y^*_{\lambda\mu}(\theta_y,\varphi_y)\phi_l^*(y)\phi_m(y)d\theta_yd\varphi_y\right)\bigg\}dr_xdr_y,
\end{split}
\end{equation}
Since $\int \phi_j(x)\phi^*_k(x)dx=0$ follows from $j\neq k$, the term corresponding to $Y_{00}=\frac{1}{\sqrt{4\pi}}$ in \eqref{myeq4.1} vanishes. The higher order terms would also be small, because the condition $\{j,k\}\neq \{l,m\}$ means that $\phi_j\phi^*_k\neq \phi_l^*\phi_m$ and the main components of $\phi_j\phi^*_k$ and $\phi_l^*\phi_m$ with respect to the spherical harmonics would be different. When $l\neq m$, the estimate is obtained in a similar way.

Finally the factor $u_{jl}^{-1}$ in (LMO) (iii) comes from the factor $\frac{1}{|x-y|}$ and the localization of $\phi_j^*(x)$ and $\phi_l^*(y)$. Note here that $\sum_lu_{jl}^{-1}$ would depend on the size of the molecule, even if the shape of the molecule is linear. The meaning of (NI) can be seen by a similar way as in that of (LMO).

\section{Application to structures of connected molecules}\label{fifthsec}
In this section we illustrate how the existence of a critical point of the Hartree-Fock functional close to an approximate critical point leads to the existence of an electronic structure of a connected molecule preserving the electronic structures of molecular fragments under several assumptions. The content of this section is not a complete proof of the assumptions (LMO) (OI) and (NI) in a particular situation. Actually, we see just that the assumptions are plausible.
\subsection{Molecular orbitals of a connected molecule}

Let us consider $M$ molecules. Assume that the molecule $\alpha=1,\dots,M$ has $n_{\alpha}$ nuclei and $\tilde N_{\alpha}$ electrons. Let $\bar x_1^{\alpha},\dots,\bar x^{\alpha}_{n_{\alpha}}$ (resp., $Z_1^{\alpha},\dots,Z_{n_{\alpha}}^{\alpha}$) be the positions (resp., atomic numbers) of $n_{\alpha}$ nuclei of the molecule $\alpha$. We consider a discretized problem of the molecule $\alpha$ by a basis $\{\chi^{\alpha}_1,\dots,\chi^{\alpha}_{\nu_{\alpha}}\}$, where $\chi_j^{\alpha}\in H^2(\mathbb R^3\times \mathcal S;\mathbb C)\subset L^2(\mathbb R^3\times \mathcal S;\mathbb C)$ is centered at one of $\{\bar x_1^{\alpha},\dots,\bar x^{\alpha}_{n_{\alpha}}\}$. In other words, we consider critical points of the Hartree-Fock functional $E^{\alpha}(\Phi^{\alpha})$ of the molecule $\alpha$ restricted to $\mathcal L(\chi^{\alpha}_1,\dots,\chi^{\alpha}_{\nu_{\alpha}})$, where $E^{\alpha}(\Phi^{\alpha})$ is the functional for $\tilde N_{\alpha}$ electrons in which nuclear positions (resp., atomic numbers) are $\bar x_1^{\alpha},\dots,\bar x^{\alpha}_{n_{\alpha}}$ (resp., $Z^{\alpha}_1,\dots,Z^{\alpha}_{n_{\alpha}}$). Let us call $\chi^{\alpha}_1,\dots,\chi^{\alpha}_{\nu_{\alpha}}$ atomic orbitals.

Let $\tilde\Phi^{\alpha}=(\psi^{\alpha}_1,\dots,\psi^{\alpha}_{\tilde N_{\alpha}}),\ \langle\psi_j^{\alpha},\psi_k^{\alpha}\rangle=\delta_{jk}$ be a solution to the corresponding Hartree-Fock equation
$$\mathcal F^{\alpha}(\tilde \Phi^{\alpha})\psi^{\alpha}_j=\epsilon^{\alpha}_j\psi_j^{\alpha},\ j=1\dots,\tilde N_{\alpha}.$$
for the molecule $\alpha$ containing $\tilde N_{\alpha}$ electrons, where $\mathcal F^{\alpha}(\Phi^{\alpha})$ is the Fock operator associated with $E^{\alpha}(\Phi^{\alpha})$. We assume that there exists a unitary matrix $R=(r_{jk})$ such that $\varphi_j^{\alpha}:=\sum_{k=1}^{\tilde N_{\alpha}}r_{jk}\psi^{\alpha}_k,\ j=1,\dots,\tilde N_{\alpha}$ are localized molecular orbitals (regarded as occupied orbitals) in the sense that estimates corresponding to (LMO), (OI) (ii) and (NI) hold. We also assume that adding localized orthonormal functions $\{\varphi^{\alpha}_{\tilde N_{\alpha}+1},\dots,\varphi^{\alpha}_{\nu_{\alpha}}\}$ (regarded as unoccupied orbitals) to $\{\varphi^{\alpha}_1,\dots,\varphi^{\alpha}_{\tilde N_{\alpha}}\}$ we can choose an orthonormal basis $\{\varphi^{\alpha}_1,\dots,\varphi^{\alpha}_{\tilde N_{\alpha}},\varphi^{\alpha}_{\tilde N_{\alpha}+1},\dots,\varphi^{\alpha}_{\nu_{\alpha}}\}$ of the linear subspace $\mathcal L(\chi^{\alpha}_1,\dots,\chi^{\alpha}_{\nu_{\alpha}})$ of $L^2(\mathbb R^3\times \mathcal S;\mathbb C)$ that also satisfy (LMO), (OI) and (NI). We associate each $\varphi_j^{\alpha}$ with a point $q^{\alpha}_j\in\mathbb R^3$ which is regarded intuitively as the center of distribution of the electron in the state $\varphi_j^{\alpha}$.
\begin{figure}[h]
\centering
\includegraphics[width=0.9\textwidth]{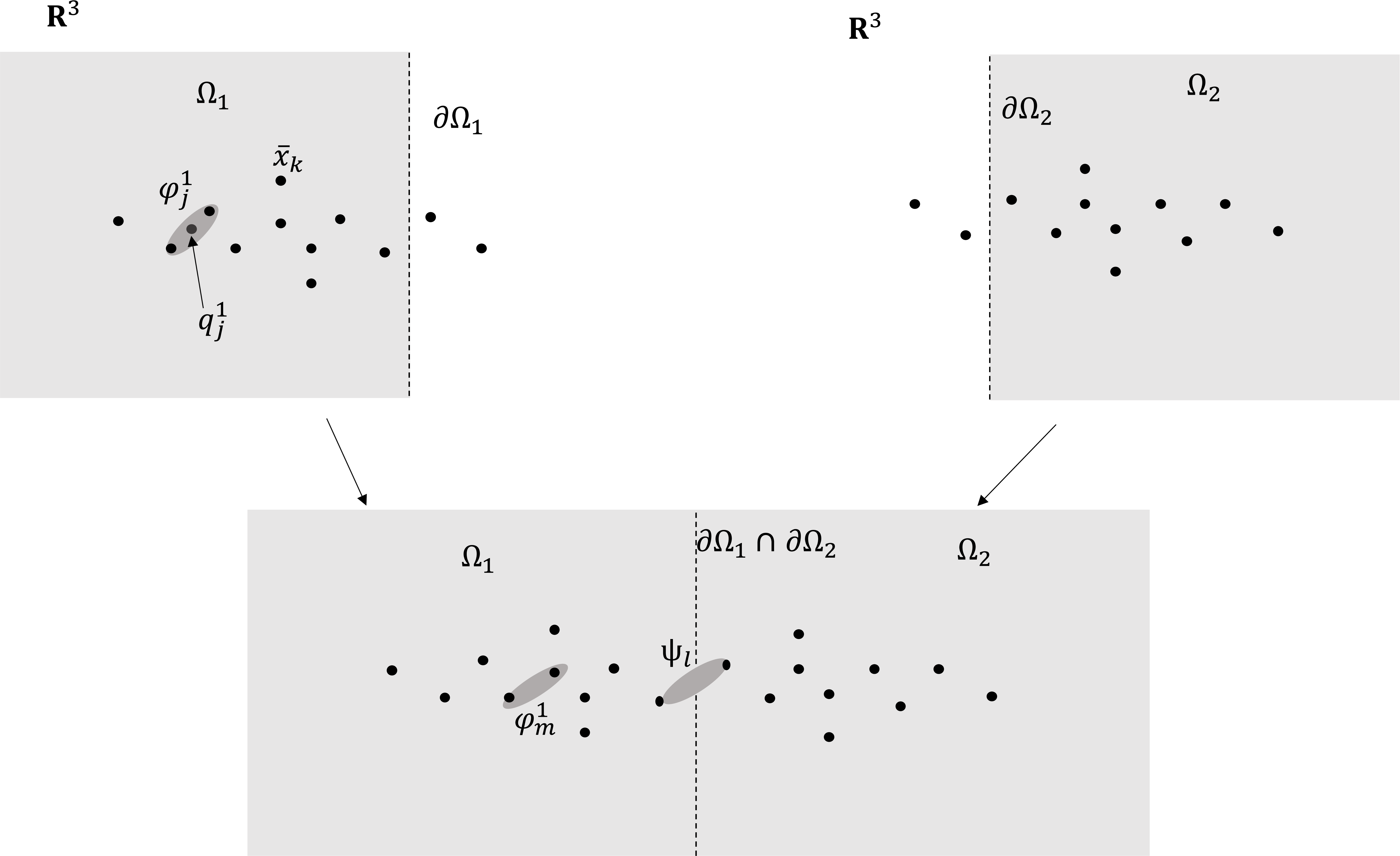}
\caption{Schematic diagram of molecular connection}\label{molfig}
\end{figure}

We assume that there exist simply connected open sets $\Omega_{\alpha}\subset \mathbb R^3,\ 1\leq \alpha\leq M$ such that $\bigcup_{\alpha=1}^M\bar \Omega_{\alpha}=\mathbb R^3$, $\Omega_{\alpha}\cap \Omega_{\beta}=\emptyset,\ \alpha\neq \beta$ and $\partial \Omega_{\alpha}\cap\{\bar x^{\alpha}_1,\dots,\bar x^{\alpha}_{n_{\alpha}}\}=\emptyset$. We consider the molecule whose nuclear positions are $\bigcup_{\alpha=1}^M(\Omega_{\alpha}\cap\{\bar x^{\alpha}_1,\dots,\bar x^{\alpha}_{n_{\alpha}}\})$ (see Figure \ref{molfig}). We relabel the nuclear positions by $1,\dots, n$ as $\{\bar x_1,\dots,\bar x_n\}$. We use the basis functions $\varphi^{\alpha}_j$ whose reference point $q_j^{\alpha}$ is in $\Omega_{\alpha}$ and away from $\partial \Omega_{\alpha}$ directly for the basis set of the combined molecule. However, the functions that contribute the chemical bonds between the molecular fragments would concentrate around $q_j^{\alpha}$ near $\partial \Omega_{\alpha}$. Therefore, we need to reconstruct such functions from the remaining basis functions. Let $\{\varphi^{\alpha}_1,\dots,\varphi^{\alpha}_{\kappa_{\alpha}}\},\ \kappa_{\alpha}<\nu_{\alpha}$ be the basis functions with $q_j^{\alpha}$ in $\Omega_{\alpha}$ and away from $\partial \Omega_{\alpha}$.
 We label $\varphi_j^{\alpha}$ so that $\{\varphi^{\alpha}_1,\dots,\varphi^{\alpha}_{\hat N_{\alpha}}\},\ \hat N_{\alpha}<\kappa_{\alpha}$ are occupied orbitals of the original molecule $\alpha$ away from $\partial\Omega_{\alpha}$.
 
 We assume that the connected molecule contains $N>\hat N:=\sum_{\alpha=1}^M\hat N_{\alpha}$ electrons and that there exist orthonormal functions $\{\psi_1,\dots,\psi_{N-\hat N}\}$ localized near $\bigcup_{\alpha=1}^M\partial \Omega_{\alpha}$ such that $\mathcal L(\psi_1,\dots,\psi_{N-\hat N})\perp \left(\sum_{\alpha=1}^M\mathcal L(\varphi^{\alpha}_1,\dots,\varphi^{\alpha}_{\kappa_{\alpha}})\right)$, and that $\{\psi_1,\dots\newline ,\psi_{N-\hat N}\}\cup\left(\bigcup_{\alpha=1}^M\{\varphi^{\alpha}_1,\dots,\varphi^{\alpha}_{\hat N_{\alpha}}\}\right)$ is close to a critical point of $E(\Phi)$ in $\bigoplus_{j=1}^NL^2(\mathbb R^3\times S;\mathbb C)$, where $E(\Phi)$ is the Hartree-Fock functional for the nuclear positions $\{\bar x_1,\dots\newline ,\bar x_n\}$. (Note that in general $\langle \varphi^{\alpha}_j,\varphi^{\beta}_k\rangle\neq0$ for $\alpha\neq \beta$, and thus the functions do not satisfy the orthnormal constraint of the Hartree-Fock equation.) Assume also that there exist orthonormal functions $\{\tilde \psi_{1},\dots,\tilde\psi_{\lambda}\},\ \lambda\in\mathbb N$ localized near $\bigcup_{\alpha=1}^M\partial \Omega_{\alpha}$ such that $\mathcal L(\tilde \psi_{1},\dots,\tilde\psi_{\lambda})\perp\left\{\mathcal L(\psi_1,\dots,\psi_{N-\hat N})+\left(\sum_{\alpha=1}^N\mathcal L(\varphi^{\alpha}_1,\dots,\varphi^{\alpha}_{\hat N_{\alpha}})\right)\right\}$, and that atomic orbitals $\chi^{\alpha}_j$ whose reference points $q_j^{\alpha}$ are in $\Omega_{\alpha}$ or close to the boundary $\partial \Omega_{\alpha}$ are included in $\mathcal L\{\tilde \psi_{1},\dots,\tilde\psi_{\lambda}, \psi_1,\dots,\psi_{N-\hat N}\}+\left(\sum_{\alpha=1}^N\mathcal L(\varphi^{\alpha}_1,\dots,\varphi^{\alpha}_{\hat N_{\alpha}})\right)$.
Intuitively, $\{\psi_1,\dots,\psi_{N-\hat N}\}$ are occupied orbitals and $\{\tilde \psi_{1},\dots,\tilde\psi_{\lambda}\}$ are unoccupied orbitals. If the structures of the molecules $\alpha$ and $\beta$ are similar near $\partial\Omega_{\alpha}\cap\partial\Omega_{\beta}$, we would be able to choose $\psi_j$ and $\tilde\psi_k$ close to molecular orbitals of the molecule $\alpha$ or $\beta$ there. The properties and positions of the molecular orbitals are tabulated in Table \ref{motable}.

\begin{table}[h]
    \centering
    \caption{properties and positions of the molecular orbitals}
    \label{motable}
    \begin{tabular}{|c|c|c|}
        \hline
         & property & position\\ \hline
        $\varphi^{\alpha}_1,\dots,\varphi^{\alpha}_{\hat N_{\alpha}}$ & occupied & away from $\bigcup_{\alpha=1}^M\partial \Omega_{\alpha}$ \\ \hline
        $\varphi^{\alpha}_{\hat N_{\alpha}+1},\dots,\varphi_{\kappa_{\alpha}}^{\alpha}$ & unoccupied & away from $\bigcup_{\alpha=1}^M\partial \Omega_{\alpha}$ \\ \hline
        $\psi_1,\dots,\psi_{N-\hat N}$ & occupied & near $\bigcup_{\alpha=1}^M\partial \Omega_{\alpha}$ \\ \hline
        $\tilde \psi_{1},\dots,\tilde\psi_{\lambda}$ & unoccupied & near $\bigcup_{\alpha=1}^M\partial \Omega_{\alpha}$ \\ \hline
    \end{tabular}
\end{table}

If we relabel $\left(\bigcup_{\alpha=1}^M\{\varphi^{\alpha}_1,\dots,\varphi_{\hat N_{\alpha}}^{\alpha}\}\right)\cup\{\psi_1,\dots,\psi_{N-\hat N}\}$ by $1,\dots,N$ and denote it by $\Phi'=(\varphi'_1,\dots,\varphi'_N)$, from the assumptions we would be able to suppose that $\{\psi_1,\dots,\psi_{N-\hat N}\}$ and $\{\tilde \psi_{1},\dots,\tilde\psi_{\lambda}\}$ belong approximately to subspaces corresponding to different spectra of $\mathcal F(\Phi')$ separated by a gap, where $\mathcal F(\Phi)$ is the Fock operator associated with $E(\Phi)$.
We also relabel $\left(\bigcup_{\alpha=1}^M\{\varphi^{\alpha}_{\hat N_{\alpha}+1},\dots,\varphi_{\kappa_{\alpha}}^{\alpha}\}\right)\cup\{\tilde\psi_{1},\dots,\tilde\psi_{\lambda}\}$ as $\{\varphi'_{N+1},\dots,\varphi'_{\nu}\}$. We  denote by $q_j$ the reference point of the distribution of $\varphi_j'$. From the construction we can see that $\langle\varphi'_j,\varphi'_k\rangle=0$ unless $\varphi'_j=\varphi_{l}^{\alpha}$ and $\varphi'_k=\varphi_{m}^{\beta}$ for some $\alpha\neq \beta$, $l$ and $m$. When this holds, the reference points of $\varphi'_j$ and $\varphi'_k$ are apart from each other, and $\langle\varphi'_j,\varphi'_k\rangle$ is small. Thus we would be able to assume
\begin{equation}\label{myeq5.0.0.0}
\max_j\sum_{k\neq j}w_{jk}|\langle\varphi'_j,\varphi'_k\rangle|\leq \epsilon_0,
\end{equation}
for some small $\epsilon_0>0$.

(LMO) and (NI) with $\phi_j$ (resp., $\Phi^0$) replaced by $\varphi'_j$ (resp., $\Phi'$) would be plausible, since $\{\varphi_1^{\alpha},\dots,\varphi_{N_{\alpha}}^{\alpha}\}$ are localized for $\alpha=1,\dots,M$ and $\{\psi_1,\dots,\psi_{N-\hat N}\}$ and $\{\tilde \psi_{1},\dots,\tilde\psi_{\lambda}\}$ are localized near the boundaries $\partial\Omega_{\alpha},\ \alpha=1\dots,M$.
As for the condition (OI) (i) and (iii), $\{\psi_1,\dots,\psi_{N-\hat N}\}$ (resp., $\{\tilde \psi_{1},\dots,\tilde\psi_{\lambda}\}$) are approximate occupied (resp., unoccupied) orbitals corresponding to $\mathcal F(\Phi')$, and $\{\varphi^{\alpha}_1,\dots,\varphi_{\hat N_{\alpha}}^{\alpha}\}$ (resp., $\{\varphi^{\alpha}_{\hat N_{\alpha}+1},\dots,\varphi_{\kappa_{\alpha}}^{\alpha}\}$) are occupied (resp., unoccupied) orbitals corresponding to $\mathcal F^{\alpha}(\hat \Phi^{\alpha})$, where $\hat\Phi^{\alpha}:=(\varphi^{\alpha}_1,\dots,\varphi_{\hat N_{\alpha}}^{\alpha})$. Thus if we assume $\mathcal F(\Phi')$ is close to $\mathcal F^{\alpha}(\hat\Phi^{\alpha})$ as a differential operator in a certain region in $\Omega_{\alpha}$ and away from $\partial\Omega_{\alpha}$, the functions $\{\varphi'_{N+1},\dots,\varphi'_{\nu}\}$ and $\{\varphi'_1,\dots,\varphi'_N\}$ would be separated with respect to the spectrum of $\mathcal F(\Phi')$, and therefore, (OI) (i) and (iii) would hold.
The constant $\delta$ in (OI) (ii) would be somewhat small because of the localization of $\phi_j$.

\subsection{orthogonalization of the molecular orbitals}
Density matrices are orthogonal projections only when the basis $\{\phi_j\}$ is an orthonormal basis. Therefore, we need to construct an orthonormal basis $\{\phi_j\}$ from $\{\varphi'_j\}$. Here as in the argument above \eqref{myeq5.0.0.0} we can expect that $\{\varphi_j'\}$ is almost orthogonal.
We construct $\{\phi_j\}$ by the Schmidt orthogonalization, that is, assuming that we have constructed $\phi_1,\dots,\phi_k$ from $\varphi'_1,\dots,\varphi'_k$, we construct $\phi_{k+1}$ by $\tilde\varphi_{k+1}:=\varphi_{k+1}'-\sum_{l=1}^k\langle \phi_l,\varphi'_{k+1}\rangle\phi_l$ and $\phi_k:=\lVert\tilde\varphi_{k+1}\rVert^{-1}\tilde\varphi_{k+1}$.

Since $\mathcal L(\phi_1,\dots,\phi_{j-1})=\mathcal L(\varphi'_1,\dots,\varphi'_{j-1})$, the Schmidt orthogonalization can also be written as
\begin{align*}
\phi_1&:=\varphi'_1,\\
\phi_j&:=\lVert \tilde\varphi_j\rVert^{-1}\tilde\varphi_j,\ 2\leq j\leq \nu,\\
\tilde\varphi_j&:=\varphi'_j-\sum_{1\leq k,l\leq j-1}\{(A_{j-1}^{-1})_{kl}\langle\varphi'_{l},\varphi'_j\rangle\}\varphi'_{k},
\end{align*}
where $A_{j-1}$ is the $(j-1)\times(j-1)$ matrix defined by $(A_{j-1})_{kl}:=\langle\varphi'_{k},\varphi'_{l}\rangle,\ 1\leq k,l\leq j-1$ and $(A_{j-1}^{-1})_{kl}$ denotes the $(k,l)$-component of $A_{j-1}$ (Note that the orthogonal projection of $\varphi$ onto $\mathcal L(\varphi'_1,\dots,\varphi'_{j-1})$ is given by $\sum_{1\leq k,l\leq j-1}\{(A_{j-1}^{-1})_{kl}\langle\varphi'_{l},\varphi\rangle\}\varphi'_{k}$.)
Then by \eqref{myeq5.0.0.0} $A_{j-1}$ is decomposed as $A_{j-1}=I_{j-1}+\tilde A_{j-1}$, where $I_{j-1}$ is the $(j-1)\times(j-1)$ identity matrix and $\lVert \tilde A_{j-1}\rVert_{w, 1,\infty}\leq\epsilon_0$. Here the norm $\lVert\cdot\rVert_{w,1,\infty}$ is defined by
$$\lVert A\rVert_{w,1,\infty}:=\lVert (w_{jk}a_{jk})\rVert_{1,\infty}=\max \left\{\max_j\sum_kw_{jk}|a_{jk}|,\max_k\sum_jw_{jk}|a_{jk}|\right\},$$
for $A=(a_{jk})$. We have $\lVert B\tilde B\rVert_{w,1,\infty}\leq \lVert B\rVert_{w,1,\infty}\lVert \tilde B\rVert_{w,1,\infty}$ because of $w_{jl}\leq w_{jk}w_{kl}$ which follows from (W) (ii).

Let us define $\nu\times \nu$ matrices $\tilde{A'}=(\tilde \alpha_{jk})$  and $A'=(\alpha_{jk})$ by
$$\tilde{\alpha}_{kl}=\begin{cases}
|\langle \varphi'_k,\varphi_l'\rangle| &k\neq l\\
0& k=l
\end{cases},$$
and by $A':=\sum_{m=0}^{\infty}\tilde{(A')}^m$ respectively. We can easily see that by \eqref{myeq5.0.0.0} the right-hand side converges and $\lVert A'\rVert_{w,1,\infty}\leq(1-\epsilon_0)^{-1}$. Using the Neumann series $A_{j-1}^{-1}=\sum_{m=0}^{\infty}(-1)^m\tilde A_{j-1}^m$ for $A_{j-1}^{-1}=(I_{j-1}+\tilde A_{j-1})^{-1}$ and considering the expression of the components of products $\tilde A_{j-1}^m$ of matrices by the components of $\tilde A_{j-1}$, we can easily see that $|(A_{j-1}^{-1})_{kl}|\leq \alpha_{kl}$ for $1\leq k,l\leq j-1$.
Thus we have
$$\lVert\tilde\varphi_j-\varphi'_j\rVert\leq \sum_{1\leq k,l\leq j-1}|(A_{j-1}^{-1})_{kl}\langle\varphi'_{l},\varphi'_j\rangle|\leq \sum_{1\leq k,l\leq j-1}|\alpha_{kl}\langle\varphi'_{l},\varphi'_j\rangle|\leq\epsilon_0(1-\epsilon_0)^{-1},$$
and therefore,
$$1-\epsilon_1\leq\lVert\tilde\varphi_j\rVert\leq 1+\epsilon_1,$$
where $\epsilon_1:=\epsilon_0(1-\epsilon_0)^{-1}$.

We define an $\nu\times \nu$ matrix $S=(s_{jk})$ by
$$s_{jk}=\begin{cases}
\lVert \tilde\varphi_j\rVert^{-1}\sum_{l=1}^{j-1}(A_{j-1}^{-1})_{kl}\langle \varphi'_l,\varphi'_j\rangle &\quad k\leq j-1\\
0 &\quad \mathrm{otherwise}
\end{cases}.$$
Then we have
\begin{equation}\label{myeq5.0}
\phi_j:=\lVert\tilde\varphi_j\rVert^{-1}\varphi'_j-\sum_{k=1}^{j-1}s_{jk}\varphi'_{k},
\end{equation}
and
\begin{align*}
&\max_j\sum_kw_{jk}|s_{jk}|\\
&\quad\leq(1-\epsilon_1)^{-1}\max_j\sum_k\sum_{l<j}w_{jk}|(A_{j-1}^{-1})_{kl}||\langle\varphi'_l,\varphi'_j\rangle|\\
&\quad\leq(1-\epsilon_1)^{-1}\max_j\sum_k\sum_{l<j}w_{jk}|\alpha_{kl}||\langle\varphi'_l,\varphi'_j\rangle|\\
&\quad\leq(1-\epsilon_1)^{-1}\max_j\sum_k\sum_{l<j}(w_{jk}w_{kl}^{-1}w_{lj}^{-1})(w_{kl}|\alpha_{kl}|)(w_{lj}|\langle\varphi'_l,\varphi'_j\rangle|)\\
&\quad\leq (1-\epsilon_1)^{-1}\lVert A'\rVert_{w,1,\infty}\epsilon_0\leq\epsilon_0(1-\epsilon_0)^{-1}(1-\epsilon_1)^{-1}=\epsilon_2,
\end{align*}
where $\epsilon_2:=\epsilon_0(1-\epsilon_0)^{-1}(1-\epsilon_1)^{-1}$.
In the same way we obtain $\max_k\sum_jw_{jk}|s_{jk}|\leq\epsilon_2$, and therefore, $\lVert S\rVert_{w,1,\infty}\leq\epsilon_2$.
Now let us prove (LMO), (OI) and (NI) for $\phi_j$ using the conditions for $\varphi_j'$. We denote the quantities for $\varphi_j'$ by putting primes as $\tilde C'$, $u_{jl}'$, $\Lambda'_{jklm}:=\int(\varphi'_j)^*(x)\varphi'_k(x)\frac{1}{|x-y|}\varphi'_l(y)\varphi'_m(y)dxdy$,  etc.

\begin{proposition}\label{philmo}
Let $\{\phi_j\}$ be the orthonormal system constructed as above. Then we have (LMO) (i)--(vi) with the constants
\begin{align*}
\tilde\epsilon&:=\tilde \epsilon'(1-\epsilon_1)^{-4}+\epsilon_3\tilde C',\\
\tilde C&:=\{(1-\epsilon_1)^{-4}+\epsilon_3\}\tilde C',\\
\hat C&:=\{(1-\epsilon_1)^{-4}+\epsilon_3\}\hat C',\\
\check C&:=\{(1-\epsilon_1)^{-2}+(2\epsilon_2(1-\epsilon_1)^{-1}+\epsilon_2^2)\}\check C',
\end{align*}
(OI) (i)--(iii) with the constants
\begin{align*}
\epsilon:=&(1-\epsilon_1)^{-2}\epsilon'+(2\epsilon_2(1-\epsilon_1)^{-1}+\epsilon_2^2)(\check C'+\breve C')+2\epsilon_4\tilde\epsilon'(1-\epsilon_1)^{-2}+\epsilon_3(\tilde C'+\hat C'),\\
\delta:=&(1-\epsilon_1)^{-2}\delta'+(2\epsilon_2(1-\epsilon_1)^{-1}+\epsilon_2^2)(\check C'+\breve C')+2\epsilon_4\tilde\epsilon'(1-\epsilon_1)^{-2}+\epsilon_3(\tilde C'+\hat C'),\\
\gamma:=&(1+\epsilon_1)^{-2}\gamma'-2(2\epsilon_2(1-\epsilon_1)^{-1}+\epsilon_2^2)(\check C'+\breve C')-2\epsilon_4(\tilde C'+2\hat C')(1-\epsilon_1)^{-2}\\
&-2\epsilon_3(\tilde C'+\hat C'),
\end{align*}
and (NI) with the constant
$$\breve C:=\{(1-\epsilon_1)^{-2}+2\epsilon_2(1-\epsilon_1)^{-1}+\epsilon_2^2\}\breve C',$$
where $\epsilon_3:=4\epsilon_2(1-\epsilon_1)^{-3}+6\epsilon_2^2(1-\epsilon_1)^{-2}+4\epsilon_2^3(1-\epsilon_1)^{-1}+\epsilon_2^4$ and $\epsilon_4:=\max\{(1-\epsilon_1)^{-2}-1,1-(1+\epsilon_1)^{-2}\}$.
\end{proposition}

\begin{remark}
The constants $\epsilon_k,\ k=1,\dots,4$ are small for small $\epsilon_0$.
\end{remark}

We shall prove Proposition \ref{philmo} in the order of (LMO), (OI), (NI).
\begin{proof}[proof of (LMO) for $\phi_j$]
(i) By \eqref{myeq5.0} we have
\begin{equation}\label{myeq5.1}
[jk|lm]=(\lVert\tilde\varphi_j\rVert\lVert\tilde\varphi_k\rVert\lVert\tilde\varphi_l\rVert\lVert\tilde\varphi_m\rVert)^{-1}\Lambda'_{jklm}+\sum_{\tau=1}^4I_{\tau},
\end{equation}
where $I_{\tau}$ is the sum of terms including $\tau$ factors like $s_{jk}$. For example,
\begin{align*}
|I_4|&=\left|\sum_{\tilde j,\tilde k,\tilde l,\tilde m}s_{j\tilde j}s_{k\tilde k}\Lambda'_{\tilde j\tilde k\tilde l\tilde m}s_{l\tilde l}s_{m\tilde m}\right|\\
&\leq\sum_{\tilde j,\tilde k,\tilde l,\tilde m}w_{\tilde j\tilde k}^{-1}(u_{\tilde j\tilde l}')^{-1}w_{\tilde l\tilde m}^{-1}|s_{j\tilde j}||s_{k\tilde k}||s_{l\tilde l}||s_{m\tilde m}|\\
&\leq\sum_{\tilde j,\tilde k,\tilde l,\tilde m}(w_{k\tilde k}^{-1}w_{\tilde j\tilde k}^{-1}w_{j\tilde j}^{-1})(u_{\tilde j\tilde l}')^{-1}(w_{l\tilde l}^{-1}w_{\tilde l\tilde m}^{-1}w_{m\tilde m}^{-1})\\
&\hspace{150pt}\cdot|w_{j\tilde j}s_{j\tilde j}||w_{k\tilde k}s_{k\tilde k}||w_{l\tilde l}s_{l\tilde l}||w_{m\tilde m}s_{m\tilde m}|\\
&\leq w_{kj}^{-1}w_{lm}^{-1}\hat u_{jl}^{-1}\epsilon_2^2,
\end{align*}
where $\hat u_{jl}^{-1}:=\sum_{\tilde j,\tilde l}(u_{\tilde j\tilde l}')^{-1}|w_{j\tilde j}s_{j\tilde j}||w_{l\tilde l}s_{l\tilde l}|=(S_wU'S_w^T)_{jl}$. Here matrices are defined by $U':=((u'_{jk})^{-1})$ and $S_w:=(|w_{jk}s_{jk}|)$. We have $\sum_j\hat u_{jl}^{-1}\leq \lVert S\rVert_{w,1,\infty}^2\sum_{\tilde j}(u_{\tilde j\tilde l}')^{-1}\leq\epsilon_2^2\tilde C'$. 
In the same way we can see that $|I_1|\leq 4\epsilon_2(1-\epsilon_1)^{-3}\tilde C'w_{jk}^{-1}\check u_{jl}^{-1}w_{lm}^{-1}$, $|I_2|\leq 6\epsilon_2^2(1-\epsilon_1)^{-2}\tilde C'w_{jk}^{-1}\check u_{jl}^{-1}w_{lm}^{-1}$ and $|I_3|\leq 4\epsilon_2^3(1-\epsilon_1)^{-1}\tilde C'w_{jk}^{-1}\check u_{jl}^{-1}w_{lm}^{-1}$, where
$$\check u_{jl}^{-1}:=\max\{ (\tilde C')^{-1}(u_{jl}')^{-1}, \epsilon_2^{-2}(\tilde C')^{-1}\hat u_{jl}^{-1}, \epsilon_2^{-1}(\tilde C')^{-1}(S_wU')_{jl},\epsilon_2^{-1}(\tilde C')^{-1}(U'S_w^T)_{jl}\},$$
and $\sum_j\check u_{jl}^{-1}\leq 1$. Using these estimates it follows from \eqref{myeq5.1} that
$$|[jk|lm]|\leq\tilde\epsilon'(1-\epsilon_1)^{-4} (v_{jl}')^{-1}w_{jk}^{-1}w_{lm}^{-1}+\epsilon_3\tilde C'\check u_{jl}^{-1}w_{jk}^{-1}w_{lm}^{-1}.$$
Thus if we set
$$\tilde v_{jl}^{-1}:=\tilde\epsilon'(1-\epsilon_1)^{-4} (v_{jl}')^{-1}+\epsilon_3\tilde C'\check u_{jl}^{-1},$$
$\tilde\epsilon:=\max_l\sum_j\tilde v_{jl}^{-1}$ and $v_{jl}^{-1}:=(\sum_j\tilde v_{jl})^{-1}\tilde v_{jl}^{-1}$, we obtain the estimate.

(iii) In a similar way as in (i) we obtain
$$|[jk|lm]|\leq \{(1-\epsilon_1)^{-4}+\epsilon_3\}\tilde C'\check u_{jl}^{-1}w_{jk}^{-1}w_{lm}^{-1}.$$
Hence if we set $\tilde C:=\{(1-\epsilon_1)^{-4}+\epsilon_3\}\tilde C'$ and $u_{jl}^{-1}:=\{(1-\epsilon_1)^{-4}+\epsilon_3\}\tilde C'\check u_{jl}^{-1}$, we obtain the estimate.

(v) In a similar way as in (i) we obtain
$$|[jk|lm]|\leq \{(1-\epsilon_1)^{-4}+\epsilon_3\}\hat C' w_{jk}^{-1}w_{lm}^{-1}.$$
Thus if we set $\hat C:=\{(1-\epsilon_1)^{-4}+\epsilon_3\}\hat C'$ , we obtain the estimate.

(vi) We have
$$\langle \nabla\phi_j,\nabla\phi_k\rangle=(\lVert\tilde\varphi_j\rVert\lVert\tilde\varphi_k\rVert)^{-1}\langle \nabla\varphi'_j,\nabla\varphi'_k\rangle+\hat I_1+\hat I_2,$$
where $\hat I_{\tau}$ is the sum of terms including $\tau$ factors like $s_{jk}$. It is easily seen that the sum with respect to $k$ of the absolute value of the right-hand side is bounded from above by
$$\check C:=(1-\epsilon_1)^{-2}\check C'+(2\epsilon_2(1-\epsilon_1)^{-1}+\epsilon_2^2)\check C'.$$
\end{proof}

\begin{proof}[proof of (OI) for $\phi_j$]
(i) We have
\begin{equation}\label{myeq5.1.1}
\begin{split}
|\langle \phi_j,\mathcal F(\Phi)\phi_k\rangle|=&|(\lVert \tilde\varphi_j\rVert\lVert \tilde\varphi_k\rVert)^{-1}\langle \varphi'_j,\mathcal F(\Phi')\varphi'_k\rangle|\\
&+|\langle \phi_j,\mathcal F(\Phi)\phi_k\rangle-(\lVert \tilde\varphi_j\rVert\lVert \tilde\varphi_k\rVert)^{-1}\langle \varphi'_j,\mathcal F(\Phi')\varphi'_k\rangle|.
\end{split}
\end{equation}
Let us consider the second term in the right-hand side. Using (LMO) (vi) for $\varphi'_j$ the kinetic energy term of the second term is estimated as
\begin{equation}\label{myeq5.1.2}
\sum_k|\langle\nabla\phi_j,\nabla\phi_k\rangle-(\lVert \tilde\varphi_j\rVert\lVert \tilde\varphi_k\rVert)^{-1}\langle\nabla\varphi'_j,\nabla\varphi'_k\rangle|\leq (2\epsilon_2(1-\epsilon_1)^{-1}+\epsilon_2^2)\check C'.
\end{equation}
Using (NI) for $\varphi'_j$ it is easily seen that the Coulomb interaction term between the molecular orbitals and the nuclei is bounded as
\begin{equation}\label{myeq5.1.3}
\sum_k\left|\sum_l\langle \phi_j,V\phi_k\rangle-(\lVert \tilde\varphi_j\rVert\lVert \tilde\varphi_k\rVert)^{-1}\langle \varphi'_j,V\varphi'_k\rangle\right|\leq (2\epsilon_2(1-\epsilon_1)^{-1}+\epsilon_2^2)\breve C'.
\end{equation}

The term $\sum_l(\Lambda_{jkll}-(\lVert \tilde\varphi_j\rVert\lVert \tilde\varphi_k\rVert)^{-1}\Lambda'_{jkll})$ is decomposed as follows.
\begin{equation}\label{myeq5.2}
\begin{split}
&\sum_l(\Lambda_{jkll}-(\lVert \tilde\varphi_j\rVert\lVert \tilde\varphi_k\rVert)^{-1}\Lambda'_{jkll})\\
&\quad=\sum_l(\Lambda_{jkll}-(\lVert \tilde\varphi_j\rVert\lVert \tilde\varphi_k\rVert\lVert \tilde\varphi_l\rVert^2)^{-1}\Lambda'_{jkll})+\sum_l(\lVert \tilde\varphi_l\rVert^{-2}-1)(\lVert \tilde\varphi_j\rVert\lVert \tilde\varphi_k\rVert)^{-1}\Lambda'_{jkll}.
\end{split}
\end{equation}
Here we note that $|\lVert \tilde\varphi_l\rVert^{-2}-1|\leq \epsilon_4$. Since $j\neq k$, by (LMO) (i) for $\varphi'_j$ the second term in the right-hand side of \eqref{myeq5.2} is estimated as
\begin{equation}\label{myeq5.3}
\begin{split}
&\sum_k\left|\sum_l(\lVert \tilde\varphi_l\rVert^{-2}-1)(\lVert \tilde\varphi_j\rVert\lVert \tilde\varphi_k\rVert)^{-1}\Lambda'_{jkll}\right|\\
&\quad\leq\sum_k\sum_l\epsilon_4(1-\epsilon_1)^{-2}\tilde\epsilon'(v_{jl}')^{-1}w_{jk}^{-1}\leq\epsilon_4\tilde\epsilon'(1-\epsilon_1)^{-2}.
\end{split}
\end{equation}
The first term in the right-hand side of \eqref{myeq5.2} is written as a sum of terms $\check I_{\tau},\ 1\leq \tau\leq 4$ in which $\tau$ factors like $\varphi'_j$ in $\Lambda'_{jkll}$ are replaced by those like $\sum_{\tilde j}s_{j\tilde j}\varphi'_{\tilde j}$. If all factors are replaced, the term can be estimated as follows. Using (LMO) (iii) for $\varphi'_j$ we have
\begin{align*}
|\check I_4|&=\sum_k\left|\sum_{l,\tilde j,\tilde k,\tilde l,\tilde m}\Lambda'_{\tilde j\tilde k\tilde l\tilde m}s_{j\tilde j}s_{k\tilde k}s_{l\tilde l}s_{l\tilde m}\right|\\
&\leq\sum_k\sum_{l,\tilde j,\tilde k,\tilde l,\tilde m} w_{\tilde j\tilde k}^{-1}(u'_{\tilde j\tilde l})^{-1}w_{\tilde l\tilde m}^{-1}|s_{j\tilde j}s_{k\tilde k}s_{l\tilde l}s_{l\tilde m}|\\
&\leq\sum_k\sum_{l,\tilde j,\tilde k,\tilde l,\tilde m}w_{\tilde j\tilde k}^{-1}(u'_{\tilde j\tilde l})^{-1}|s_{j\tilde j}||s_{k\tilde k}||s_{l\tilde l}||s_{l\tilde m}|\leq \epsilon_2^4\tilde C',
\end{align*}
where the sums are accumulated in the order of $\tilde m,l,\tilde l,k,\tilde k,\tilde j$. The estimate for the summation with respect to $j$ is similar. In similar ways we can see that $|\check I_1|\leq 4\epsilon_2(1-\epsilon_1)^{-3}\tilde C'$, $|\check I_2|\leq 6\epsilon_2^2(1-\epsilon_1)^{-2}\tilde C'$ and $|\check I_3|\leq 4\epsilon_2^3(1-\epsilon_1)^{-1}\tilde C'$. Hence we obtain
\begin{equation}\label{myeq5.4}
\left|\sum_l(\Lambda_{jkll}-(\lVert \tilde\varphi_j\rVert\lVert \tilde\varphi_k\rVert\lVert \tilde\varphi_l\rVert^2)^{-1}\Lambda'_{jkll})\right|\leq\epsilon_3\tilde C'.
\end{equation}

The term $\sum_l(\Lambda_{jllk}-(\lVert \tilde\varphi_j\rVert\lVert \tilde\varphi_k\rVert)^{-1}\Lambda'_{jllk})$ is decomposed as
\begin{equation}\label{myeq5.5}
\begin{split}
&\sum_l(\Lambda_{jllk}-(\lVert \tilde\varphi_j\rVert\lVert \tilde\varphi_k\rVert)^{-1}\Lambda'_{jllk})\\
&\quad=\sum_l(\Lambda_{jllk}-(\lVert \tilde\varphi_j\rVert\lVert \tilde\varphi_k\rVert\lVert \tilde\varphi_l\rVert^2)^{-1}\Lambda'_{jllk})+\sum_l(1-\lVert \tilde\varphi_l\rVert^{-2})(\lVert \tilde\varphi_j\rVert\lVert \tilde\varphi_k\rVert)^{-1}\Lambda'_{jllk}.
\end{split}
\end{equation}
The second term is bounded by $\epsilon_4\tilde\epsilon'(1-\epsilon_1)^{-2}$ in the same way as in \eqref{myeq5.3}. The first term is decomposed into $\sum_{\tau=1}^4\breve I_{\tau}$ in the same way as that for $\Lambda_{jkll}$. The term $\breve I_4$ is estimated as follows.
\begin{align*}
|\breve I_4|&=\sum_k\left|\sum_{l,\tilde j,\tilde k,\tilde l,\tilde m}\Lambda'_{\tilde j\tilde l\tilde m\tilde k}s_{j\tilde j}s_{k\tilde k}s_{l\tilde l}s_{l\tilde m}\right|\\
&\leq\sum_k\sum_{l,\tilde j,\tilde k,\tilde l,\tilde m} \hat C'w_{\tilde j\tilde l}^{-1}w_{\tilde m\tilde k}^{-1}|s_{j\tilde j}s_{k\tilde k}s_{l\tilde l}s_{l\tilde m}|\leq \epsilon_2^4\hat C',
\end{align*}
where the sums are accumulated in the order of $k,\tilde k,\tilde m,l,\tilde l,\tilde j$. The estimate for the summation with respect to $j$ is similar. The terms $\breve I_{\tau},\ \tau=1,2,3$  are estimated in similar ways. Hence we have
\begin{equation}\label{myeq5.6}
\left|\sum_l(\Lambda_{jllk}-(\lVert \tilde\varphi_j\rVert\lVert \tilde\varphi_k\rVert\lVert \tilde\varphi_l\rVert^2)^{-1}\Lambda'_{jllk})\right|\leq\epsilon_3\hat C'.
\end{equation}
Combining \eqref{myeq5.1.1}--\eqref{myeq5.6} we obtain
\begin{align*}
&\max_{j\in J_o}\sum_{k\in J_u}|\langle \phi_j,\mathcal F(\Phi)\phi_k\rangle|\\
&\quad\leq(1-\epsilon_1)^{-2}\epsilon'+(2\epsilon_2(1-\epsilon_1)^{-1}+\epsilon_2^2)(\check C'+\breve C')+2\epsilon_4\tilde\epsilon'(1-\epsilon_1)^{-2}+\epsilon_3(\tilde C'+\hat C').
\end{align*}
The same estimate for the summation with respect to $j$ is also obtained in the same way.
Thus if we set the right-hand side by $\epsilon$ we obtain the estimate.

The estimates of (ii) and (iii) are obtained in the same way as in (i).
\end{proof}

\begin{proof}[proof of (NI) for $\phi_j$]
In a manner similar to that in the proof of (LMO) we can see that the left-hand side of \eqref{myeq1.3} is bounded from above by
$$(1-\epsilon_1)^{-2}(\breve u_{jkl}')^{-1} w_{jk}^{-1}+(2\epsilon_2(1-\epsilon_1)^{-1}+\epsilon_2^2)\bar u_{jkl}^{-1}w_{jk}^{-1},$$
where $\bar u_{jkl}^{-1}$ is defined in a manner similar to that of $\check u_{jl}^{-1}$ in the proof of (LMO) from $(\breve u_{jkl}')^{-1}$. Thus if we set $\breve u_{jkl}^{-1}:=(1-\epsilon_1)^{-2}(\breve u_{jkl}')^{-1}+(2\epsilon_2(1-\epsilon_1)^{-1}+\epsilon_2^2)\breve C'\bar u_{jkl}^{-1}$ and $\breve C:=\max_{jk}\sum_l\breve u_{jkl}^{-1}$, we obtain the estimate.
\end{proof}

\subsection{Interpretation of the main theorem}
By Theorem \ref{mainthm} there exists a density matrix $P^{\infty}=(p^{\infty}_{jk})$ close to $P^0$. Let us consider the electronic density as a typical quantity of an electronic structure. By \eqref{myeq1.1} using a density matrix $P=(p_{jk})$ the electronic density $\rho(x)$ is given as follows (see e.g. \cite{ZO}[Section 3.4]).
$$\rho(x):=\sum_{l=1}^N|\varphi_l(x)|^2=\sum_{j,k}p_{jk}\phi_j(x)\phi_k^*(x).$$
By the localization of $\{\phi_j\}$ we may assume only a few orbitals have significant contribution to the electronic density in a small region $\tilde \Omega\subset\mathbb R^3$. Let $\phi_1,\dots,\phi_{j_0}$ be such orbitals. Then the expectation value of the number of electrons found in $\tilde \Omega$ is approximately given by
$$\int_{\tilde \Omega}\rho(x)dx\simeq \sum_{1\leq j,k\leq j_0}p_{jk}\int_{\tilde \Omega}\phi_j(x)\phi_k^*(x)dx.$$
In particular, if $j_0=1$, we have $\int_{\tilde \Omega}\rho(x)dx\simeq p_{11}\int_{\tilde \Omega}|\phi_1(x)|^2dx$. Recall also that if the reference point $q_1$ of $\phi_1$ is away from the boundary $\bigcup_{\alpha=1}^M\Omega_{\alpha}$, $\phi_1$ is close to some $\varphi^{\alpha}_j$. Therefore, since $P^{\infty}$ is close to $P^0$, the local electronic density after the molecular connection would be close to that of the original molecular fragment, when $\tilde \Omega$ is away from the boundary $\bigcup_{\alpha=1}^M\Omega_{\alpha}$.

\bigskip

\noindent\textbf{Acknowledgment}
This work was supported by JSPS KAKENHI Grant Number JP23K13030.

\end{document}